\documentclass[sigconf]{acmart}
\pdfoutput=1

\usepackage{balance}

\def\BibTeX{{\rm B\kern-.05em{\sc i\kern-.025em b}\kern-.08emT\kern-.1667em\lower.7ex\hbox{E}\kern-.125emX}}
    
\copyrightyear{2019}
\acmYear{2019}
\setcopyright{rightsretained}
\acmConference[SPAA '19]{31st ACM Symposium on Parallelism in Algorithms and
Architectures}{June 22--24, 2019}{Phoenix, AZ, USA}
\acmBooktitle{31st ACM Symposium on Parallelism in Algorithms and
Architectures (SPAA '19), June 22--24, 2019, Phoenix, AZ, USA}
\acmDOI{10.1145/3323165.3323179}
\acmISBN{978-1-4503-6184-2/19/06}

\acmSubmissionID{spaaf023}

\usepackage{standalone}
\usepackage{color}
\usepackage{bbm}
\usepackage[ruled,lined,linesnumbered,commentsnumbered]{algorithm2e}
\usepackage[english]{babel}
\usepackage{tikz}
\usetikzlibrary{calc, arrows.meta, decorations, decorations.pathmorphing, calc, patterns}

\allowdisplaybreaks

\usepackage{graphicx}
\usepackage{pgfplots}

\usepackage{thmtools}
\usepackage{thm-restate}

\usepackage{enumerate}
\usepackage{float}

\theoremstyle{plain}
\newtheorem{theorem}{Theorem}[section]

\newtheorem{lemma}[theorem]{Lemma}
\newtheorem{proposition}[theorem]{Proposition}

\theoremstyle{definition}

\theoremstyle{remark}

\newcommand{\mytitle}{Near Optimal Coflow Scheduling in Networks}
\usepackage{url}

\newif\ifjournal
\journaltrue

\def\showauthornotes{1}

\ifnum\showauthornotes=1

\else

\fi

\ifnum\showauthornotes=1
\newcommand{\Authornote}[2]{{\sf\small\color{red}{[#1: #2]}}}
\newcommand{\Authoredit}[2]{{\sf\small\color{red}{[#1]}\color{blue}{#2}}}
\newcommand{\Authorcomment}[2]{{\sf \small\color{gray}{[#1: #2]}}}
\newcommand{\Authorfnote}[2]{\footnote{\color{red}{#1: #2}}}
\newcommand{\Authorfixme}[1]{\Authornote{#1}{\textbf{??}}}
\newcommand{\Authormarginmark}[1]{\marginpar{\textcolor{red}{\fbox{
#1:!}}}}
\else
\newcommand{\Authornote}[2]{}
\newcommand{\Authoredit}[2]{}
\newcommand{\Authorcomment}[2]{}
\newcommand{\Authorfnote}[2]{}
\newcommand{\Authorfixme}[1]{}
\newcommand{\Authormarginmark}[1]{}
\fi

\title{\mytitle}
\author{Mosharaf Chowdhury}
\email{mosharaf@umich.edu}
\affiliation{University of Michigan, Ann Arbor}
\authornote{Research supported by the National Science Foundation (CNS-1563095 and CNS-1617773)}
\author{Samir Khuller}
\email{samir.khuller@northwestern.edu}
\affiliation{Northwestern University}
\authornote{Research partially supported by an Amazon research award.}
\author{Manish Purohit}
\email{mpurohit@google.com}
\affiliation{Google, Mountain View}
\author{Sheng Yang}
\email{styang@cs.umd.edu}
\affiliation{University of Maryland, College Park}
\authornotemark[2]
\author{Jie You}
\email{jieyou@umich.edu}
\affiliation{University of Michigan, Ann Arbor}
\authornotemark[1]

\renewcommand{\shortauthors}{Chowdhury and Khuller, et al.}

\begin{document}
\begin{abstract}
  The coflow scheduling problem has emerged as a popular abstraction in the last few years to study data communication problems within a data center~\cite{chowdhury2012coflow}. In this basic framework, each coflow has a set of communication demands and the goal is to schedule many coflows in a manner that minimizes the total weighted completion time. A coflow is said to complete when all its communication needs are met. This problem has been extremely well studied for the case of complete bipartite graphs that model a data center with full bisection bandwidth and several approximation algorithms and effective heuristics have been proposed recently~\cite{agarwal2018sincronia,ahmadi2017scheduling,chowdhury2018terra}.

In this work, we study a slightly different model of coflow scheduling in general graphs (to capture traffic between data centers ~\cite{jahanjou2017asymptotically,chowdhury2018terra}) and develop practical and efficient approximation algorithms for it. Our main result is a randomized 2 approximation algorithm for the single path and free path model, significantly improving prior work. In addition, we demonstrate via extensive experiments that the algorithm is practical, easy to implement and performs well in practice.
\end{abstract}


%


%
\maketitle

\section{Introduction}
\label{sec:intro}

Modern computing applications have rather intensive computational needs. Many machine learning applications require up to tens of thousands of machines and often involve processing units across multiple data centers collaborating on the same application. This collaboration is usually handled by a large-scale distributed computing framework that ideally ensures a close-to-linear speedup compared to a single machine. A crucial part of the collaboration is that large chunks of data require both inter and intra-datacenter transmissions.

For intra-datacenter transmission, a common example would be the MapReduce framework. Map workers write all intermediate results independently to several servers to guard against failure and allow possible re-calculation. These results are shuffled and sent to Reduce workers. The volume of transmission between machines is so large that it has become a major bottleneck in the performance. In addition to this challenge, multiple applications may share the same cluster, and an un-coordinated schedule of their data transmission may cause an unacceptable delay in their completion times.

Chowdhury and Stoica~\cite{chowdhury2012coflow} first introduced the abstraction of {\em coflow scheduling}, which assumes that each application consists of a set of flows, and is finished once all the flows are completed. In their framework the network between machines is modeled as a switch: the input ports of different machines on one side, and output ports on the other side. A machine can send (receive) data to (from) any other machine, but to (from) only one machine at a time (sending and receiving may happen concurrently). 
The transmission speed between all machines is uniform. This describes a ``perfect'' datacenter where networking between machines is handled by a high-speed central switch (modeled by a complete bipartite graph) connected directly to all the machines \cite{chowdhury2012coflow}. However, real world datacenters are far more complicated; direct (virtual) links between machines may exist to avoid latency, duplicate links may exist to tolerate failure, network speeds may vary widely for different machines and links, and complicated network structures may exist for a variety of reasons. To make things worse, some tasks may involve multiple datacenters around the globe, and the switch model simply cannot accurately capture the graph based  network that connects all the data centers.

For inter-datacenter transmission, distributed machine learning tasks can generate huge amounts of traffic. Due to legal or cost reasons, some datasets cannot be gathered into a single datacenter for processing. Instead, several geographically distributed datacenters work together to train a single model, and exchange local updates frequently to ensure accuracy and convergence. Though the size of a single transmission may be small considering the network bandwidth, the repeated exchange blows up the volume of transmission and makes network traffic its bottleneck.

In order to solve these problems, a slightly different model of coflow scheduling was proposed by Jahanjou et al.~\cite{jahanjou2017asymptotically}, which assumes that the underlying connection between machines is an arbitrary graph rather than a complete bipartite graph. Each node can be a machine, a datacenter or an exchange point (switch, router, etc.), and an edge between two nodes represents a physical link between the two Internet infrastructures. When some data needs to be transmitted from one node to another, it needs to be transmitted along edges. Unlike in the switch model where only one packet can be sent at each time slot, data for multiple jobs is allowed to transfer on the same link at the same time, or in other words, shared traffic on links is allowed. The total volume of data transmission on a link however is bounded by the link bandwidth\footnote{One major challenge in the switch model is the node-wise I/O speed constraint. In order to capture this in the graph model, we can replace every datacenter with a gadget of two nodes. The first node has exactly the same neighbors and edges that the original node for the datacenter has, plus links from and to the second node. The second node is only connected to the first node, and is the true source and destination for all demands involving this datacenter. By setting capacity on the links between these two nodes, we can enforce I/O limit for the whole datacenter like in the switch model.}. Jahanjou et al.~\cite{jahanjou2017asymptotically} considered the model in which data has to travel along a single specified path. In addition to this model, we also consider the \emph{free path} model which allows data to be split or merged at nodes to utilize the whole graph when transmitting the same piece of data as long as the capacity of each link is respected. This seems much more complicated in practice than a single path transmission, but modern distributed computation frameworks~\cite{chowdhury2018terra} allow this kind of fine-grained control on network routing and transfer rate, which makes the model realistic. See Figure~\ref{fig:example_of_coflow} for a brief illustration of the two models. The formal definitions come in Section~\ref{sec:problem_formulation}.

\begin{figure}[htbp]
  \centering
  \scalebox{0.6}{\begin{tikzpicture}
    [port/.style={rounded corners, rectangle, fill=cyan!20!, draw=cyan!50!},
    weight/.style={align=center, midway},
    edge/.style={<->, thick, -{Latex[length=3mm]}},
    scale=1.5]
    \def \offset {0.1};
    \def \x {2}
    \def \y {2}
    \node[port] (hk) at (0,0) {HK};
    \node[port] (la) at (\y,0) {LA};
    \node[port] (ny) at (0, \x)  {NY};
    \node[port] (fl) at (\y,\x) {FL};
    \node[port] (ba) at (0.5*\y,1.7*\x) {BA};
    \draw[<->, draw=black!80!] (hk) to node[weight, above] {4} (la);
    \draw[<->, draw=black!80!] (hk) to node[weight, above, sloped] {2} (ny);
    \draw[<->, draw=black!80!] (ny) to node[weight, above] {4} (la);
    \draw[<->, draw=black!80!] (ny) to node[weight, above] {5} (fl);
    \draw[<->, draw=black!80!] (la) to node[weight, above, sloped] {4} (fl);
    \draw[<->, draw=black!80!] (ny) to node[weight, above] {6} (ba);
    \draw[<->, draw=black!80!] (fl) to node[weight, above] {4} (ba);
\end{tikzpicture}}%
  \hspace{0.5cm}
  \scalebox{0.6}{\begin{tikzpicture}
    [port/.style={rounded corners, rectangle, fill=cyan!20!, draw=cyan!50!},
    weight/.style={align=center, midway, sloped},
    edge/.style={->, thick, -{Latex[length=3mm]}},
    scale=1.5]
    \def \offset {0.1};
    \def \x {2}
    \def \y {2}
    \node[port] (hk) at (0,0) {HK};
    \node[port] (la) at (\y,0) {LA};
    \node[port] (ny) at (0, \x)  {NY};
    \node[port] (fl) at (\y,\x) {FL};
    \node[port] (ba) at (0.5*\y,1.7*\x) {BA};
    \draw[edge, draw=red!80!]    (hk) to [] node[weight, above] {12} (la);
    \draw[edge, draw=red!80!]    (la) to [] node[weight, above] {12} (fl);
    \draw[edge, draw=green!80!, dashed]  (ny) to [] node[weight, above] {18} (ba);
\end{tikzpicture}}%
  \hspace{0.5cm}
  \scalebox{0.6}{\begin{tikzpicture}
    [port/.style={rounded corners, rectangle, fill=cyan!20!, draw=cyan!50!},
    weight/.style={align=center, midway, sloped},
    edge/.style={->, thick, -{Latex[length=3mm]}},
    scale=1.5]
    \def \offset {0.1};
    \def \x {2}
    \def \y {2}
    \node[port] (hk) at (0,0) {HK};
    \node[port] (la) at (\y,0) {LA};
    \node[port] (ny) at (0, \x)  {NY};
    \node[port] (fl) at (\y,\x) {FL};
    \node[port] (ba) at (0.5*\y,1.7*\x) {BA};
    \draw[edge, draw=red!80!]    (hk) to [] node[weight, above] {8} (la);
    \draw[edge, draw=red!80!]    (la) to [] node[weight, above] {8} (fl);
    \draw[edge, draw=red!80!]    (hk) to [] node[weight, above] {4} (ny);
    \draw[edge, draw=red!80!]    (ny) to [bend left] node[weight, above] {4} (fl);
    \draw[edge, draw=green!80!, dashed]  (ny) to [] node[weight, above] {12} (ba);
    \draw[edge, draw=green!80!, dashed]  (ny) to [bend right] node[weight, above] {6} (fl);
    \draw[edge, draw=green!80!, dashed]  (fl) to [] node[weight, above] {6} (ba);
\end{tikzpicture}}%
  \caption{Example of coflow. The first graph shows the network topologies and the bandwidth of each link. We have one coflow consisting of two flows: one from \textsf{NY} to \textsf{BA} of demand 18 (denoted with dashed, green lines), the other from \textsf{HK} to \textsf{FL} of demand 12 (denoted with solid, red lines). The second graph shows the single path model, where each flow needs to be transmitted along a given path. It also implies a schedule in this model: transmit according to the path for 3 time units, and both flows are done. The third graph shows the free path model, where each flow can be split along multiple paths as long as the capacity of edges are respected. Here both flows can share the link from \textsf{NY} to \textsf{FL} and the entire coflow finishes in 2 units of time.}
  \label{fig:example_of_coflow}
\end{figure}
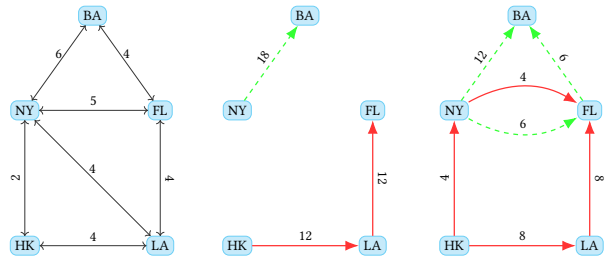

\subsection{Related Work}
\label{sec:related}
The idea of scheduling coflows was first introduced by Chowdhury and Stoica~\cite{chowdhury2012coflow}. Since then, it has been a hot topic in both the systems~\cite{luo2016towards,li2016efficient,chowdhury2014efficient,yu2016non} and the theory~\cite{qiu2015minimizing,Khuller:SPAA2016,jahanjou2017asymptotically,ahmadi2017scheduling,shafiee2018improved} communities. Most theoretical research has focused on coflow scheduling in the switch model, where the communication graph is a complete, bipartite graph. Since this basic problem generalizes  concurrent open shop scheduling and is thus NP-hard, the main results focus on the development of approximation algorithms. 
Over the last three years, a series of papers~\cite{qiu2015minimizing,Khuller:SPAA2016,ahmadi2017scheduling} have brought down the approximation factor from 67/3 to 5 for coflow scheduling with arbitrary release times and to 4 for the case without release times~\cite{ahmadi2017scheduling,shafiee2018improved} \footnote{4 is still the best known bound.}. We would like to note that a very simple primal-dual framework is proposed by Ahmadi et al.~\cite{ahmadi2017scheduling}, and this yields a very practical combinatorial algorithm for the problem without requiring the need to solve an LP (as in \cite{shafiee2018improved}). Furthermore, in recent work, a system called Sincronia~\cite{agarwal2018sincronia} was also developed based on the primal-dual method. It improves upon state-of-the-art methods and gives practical and near-optimal solutions in real testbeds.

One natural extension is to take the graph structure into consideration. Zhao et al.~\cite{zhao2015rapier} consider coflow scheduling over arbitrary graphs and attempt to jointly optimize routing and scheduling. They give a heuristic based on shortest job first, and use the idle slots to schedule flows from the longest job. Jahanjou et al.~\cite{jahanjou2017asymptotically} studied two variants of coflow scheduling over general graphs, namely, when the path for a flow is given or if the path is unspecified. In both cases, the transmission rate may change over time, but each flow can only take a single path, whether given to or chosen by the fractional routing algorithm. In the first case, Jahanjou et al.~\cite{jahanjou2017asymptotically} develop the first constant approximation algorithm (approximation ratio $17.6$) and in the second case they develop an $O(\frac{\log n}{\log\log n})$ approximation algorithm ($n$ is the number of nodes in the graph), matching the lower bound given by Chuzhoy et al.~\cite{chuzhoy2007hardness}.

Since preemption often incurs large overheads, some recent work~\cite{yu2016non} has tackled the problem of non-preemptive coflow scheduling.
Mao, Aggarwal, and Chiang~\cite{mao2018stochastic} consider the non-preemptive coflow scheduling problem with stochastic sizes and give an algorithm with an  approximation factor of $(2\log m + 1) ( 1 + \sqrt{m}\Delta)(1 + m\Delta) (3 + \Delta)/2$, where $\Delta$ is an upper bound of squared coefficient of variation of processing times. This simplifies to a $(3\log m + \frac{3}{2})$ approximation for non-stochastic cases.

\subsection{Our Contributions}
\label{sec:contributions}
The main result of this paper is a unified, tight $2$-approximation algorithm for the coflow scheduling problem in both the single path model and the free path model when all release times and demands are polynomially sized, and a $(2 + \epsilon)$-approximation when the release times and demands can be super-polynomial. This improves upon the 17.6 approximation given by Jahanjou et al.~\cite{jahanjou2017asymptotically} for the single path model, and is the first approximation algorithm for the free path model (introduced by You and Chowdhury~\cite{chowdhury2018terra}).

We also evaluated our algorithm using two WAN topologies (Microsoft's SWAN~\cite{hong2013achieving} and Google's G-Scale~\cite{jain2013b4}) on four different workloads (BigBench~\cite{big-bench}, TPC-DS~\cite{nambiar2006making}, TPC-H~\cite{poess2000new}, and Facebook (FB)~\cite{SWIM,coflow-benchmark}) and compared with state-of-the-art for both models\cite{jahanjou2017asymptotically,chowdhury2018terra}. For the single path model, we significantly improved over Jahanjou et al.~\cite{jahanjou2017asymptotically}. For the free path model, we are close to what Terra~\cite{chowdhury2018terra} gets, but have the extra capability of dealing with weights. Across all variants and models, we have shown that taking the LP solution directly is an effective heuristic in practice.

\subsection{Paper Organization}
\label{sec:org}
In Section~\ref{sec:problem_formulation} we give a formal definition of the two models for coflow scheduling. In Section~\ref{sec:lp_formulation} we give a general linear program that deals with both models. We give the additional flow constraints for the two models in Section~\ref{sec:flow constraints}. In Section~\ref{sec:algorithm} we describe the main algorithm and present the analysis in Section~\ref{sec:analysis-random}. We prove both models to be NP-hard in Section~\ref{sec:hardness}. In Section~\ref{sec:experiments}, we show experimental results by comparing our algorithms to some baseline algorithms. We conclude in Section~\ref{sec:conclusion} with some new directions to work on.

\section{Model and Problem Definition}
\label{sec:problem_formulation}
We now formally define the models of coflow scheduling that we consider in this paper. Let $G = (V,E)$ be a directed graph that represents the data center network and $c : E \rightarrow \mathbb{R}^{+}$ be a function that denotes the capacity (bandwidth) available on each edge of the network. Let $\mathcal{J} = \{F_1, F_2, \ldots, F_n\}$ denote the set of $n$ coflows. A coflow $F_j$ has weight $w_j$ that denotes its priority and consists of $n_j$ individual flows, i.e., $F_j = \{f_j^1, \ldots, f_j^{n_j}\}$ where $f_j^i = (s_j^i, t_j^i, \sigma_j^i)$ denotes a flow from source node $s_j^i \in V$ to sink $t_j^i \in V$ with demand $\sigma_j^i \in \mathbb{R}^+$. We assume that time is discrete and data transfer is instantaneous, i.e., it takes negligible time for data to cover multiple hops of edges as network delay is low compared to the time to transmit large chunks of data. A coflow $F_j$ is said to be completed at the earliest time $t$ such that for each flow $f_j^i \in F_j$, $\sigma_j^i$ units of data have been transferred from source $s_j^i$ to sink $t_j^i$. Our goal is to find a schedule that routes all the requisite flows (i.e. at any time, what fraction of a certain flow is transmitted and along which path/paths) subject to the edge bandwidth constraints so that the total weighted completion time of the coflows $\sum_j w_j C_j$ is minimized. Figure \ref{fig:coflow_setup} gives an example of an instance of the coflow scheduling problem over a simple network.

We consider {\em two} different transmission models, based on whether a flow $f_j^i$ has restrictions as to how the data is transmitted. In the {\em single path model}, each flow $f_j^i$ specifies a path $p_j^i$ from source $s_j^i \in V$ to sink $t_j^i \in V$ so that the flow can only be routed along that path. This is exactly the ``circuit-based coflows with paths given'' model studied by Jahanjou \textit{et al}.\ \cite{jahanjou2017asymptotically}.

In the {\em free path model}, we can freely select the routing we desire for any flow $f_j^i$. In any time slot, data transmission occurs as a feasible multi-commodity flow so that both flow-conservation and edge bandwidth constraints are satisfied. Thus, we can split any flow $f_j^i$ along multiple paths from its source to destination. 
This model was proposed in Terra~\cite{chowdhury2018terra}. Since the shortest paths of different flows can share edges and cause congestion, the free path model offers the flexibility of rerouting flows along less congested paths.
In addition, modern internet infrastructures support using multiple paths together to get a higher overall speed (known as link aggregation), which is captured in the free path model as network flow.

In fact, both models are handled {\em uniformly} by the same framework, and the only difference is the set of flow constraints that describe what are considered feasible transmissions. It is also possible to handle other kinds of transmissions, like an intermediate case between single path and free path: several paths are given, and we can use them together and decide at what rate we are transmitting along each path. Figures \ref{fig:singlepath} and \ref{fig:freepath} show the optimal solutions for the example coflow problem in the single path and free path models respectively.

\begin{figure}[htbp]
  \centering
  \scalebox{0.6}{\begin{tikzpicture}
    [port/.style={rounded corners, rectangle, fill=cyan!20!, draw=cyan!50!},
    weight/.style={align=center, midway},
    edge/.style={->, thick, -{Latex[length=3mm]}},
    scale=1]
    \def \offset {0.1};
    \def \x {2}
    \def \y {2}
    \node[port] (s) at (0,0) {$s$};
    \node[port] (v1) at (\x,\y) {$v_1$};
    \node[port] (v2) at (\x, 0)  {$v_2$};
    \node[port] (v3) at (\x,-\y) {$v_3$};
    \node[port] (t) at (2*\x,0) {$t$};
    \draw[edge, <->] (s) -- (v1) node[weight, above] {1};
    \draw[edge, <->] (s) -- (v2) node[weight, above] {1};
    \draw[edge, <->] (s) -- (v3) node[weight, below] {1};
    \draw[edge, <->] (v1) -- (t) node[weight, above] {1};
    \draw[edge, <->] (v2) -- (t) node[weight, above] {1};
    \draw[edge, <->] (v3) -- (t) node[weight, below] {1};
\end{tikzpicture}}%
  \hspace{1cm}
  \scalebox{0.6}{\begin{tikzpicture}
    [port/.style={rounded corners, rectangle, fill=cyan!20!, draw=cyan!50!},
    weight/.style={align=center, midway},
    edge/.style={->, thick, -{Latex[length=3mm]}},
    scale=1]
    \def \offset {0.1};
    \def \x {2}
    \def \y {2}
    \node[port] (s) at (0,0) {$s$};
    \node[port] (v1) at (\x,\y) {$v_1$};
    \node[port] (v2) at (\x, 0)  {$v_2$};
    \node[port] (v3) at (\x,-\y) {$v_3$};
    \node[port] (t) at (2*\x,0) {$t$};
    \draw[edge, draw=blue!80!, decorate, decoration=snake ] (s) to [bend left] node[weight, above] {3} (t);
    \draw[edge, draw=red!80!]    (v1) to [] node[weight, above] {1} (t);
    \draw[edge, draw=green!80!, dashed]  (v2) to [] node[weight, above] {1} (t);
    \draw[edge, draw=orange!80!, dotted] (v3) to [] node[weight, above] {1} (t);
\end{tikzpicture}}%
  \caption{On the left is the graph structure: bi-directed edge of independent capacity of $1$, on the right is the demanded coflow. There are four coflows each containing one single flow: red (solid) from $v_1$ to $t$, green (dashed) from $v_2$ to $t$, orange (dotted) from $v_3$ to $t$, and blue (curly) from $s$ to $t$. The first three have demand $1$, while the blue coflow has a demand of $3$. All of them have the same weight of $1$.}
  \label{fig:coflow_setup}
\end{figure}
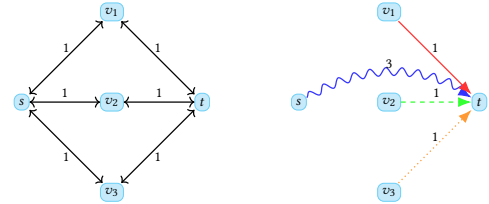
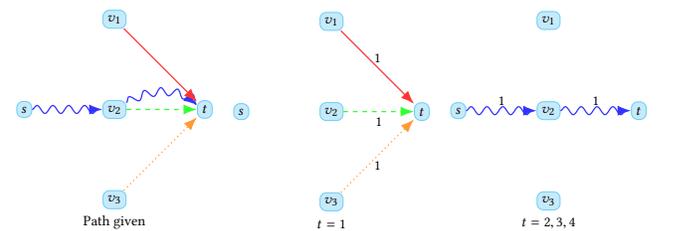
\begin{figure}[htbp]
  \centering
  \scalebox{0.6}{\begin{tikzpicture}
    [port/.style={rounded corners, rectangle, fill=cyan!20!, draw=cyan!50!},
    weight/.style={align=center, midway},
    edge/.style={->, thick, -{Latex[length=3mm]}},
    scale=1]
    \def \offset {0.5};
    \def \x {2}
    \def \y {2}
    \node[port] (s) at (0,0) {$s$};
    \node[port] (v1) at (\x,\y) {$v_1$};
    \node[port] (v2) at (\x, 0)  {$v_2$};
    \node[port] (v3) at (\x,-\y) {$v_3$};
    \node[port] (t) at (2*\x,0) {$t$};
    \draw[edge, draw=blue!80!, decorate, decoration=snake ]   (s)  to []          node[weight, above] {} (v2);
    \draw[edge, draw=blue!80!, decorate, decoration=snake ]   (v2) to [bend left] node[weight, above] {} (t);
    \draw[edge, draw=red!80!]            (v1) to []          node[weight, above] {} (t);
    \draw[edge, draw=green!80!, dashed]  (v2) to []          node[weight, below] {} (t);
    \draw[edge, draw=orange!80!, dotted] (v3) to []          node[weight, below] {} (t);
    \node () at (\x, -\y - \offset) {Path given};
\end{tikzpicture}}%
  \hspace{0.2cm}
  \scalebox{0.6}{\begin{tikzpicture}
    [port/.style={rounded corners, rectangle, fill=cyan!20!, draw=cyan!50!},
    weight/.style={align=center, midway},
    edge/.style={->, thick, -{Latex[length=3mm]}},
    scale=1]
    \def \offset {0.5};
    \def \x {2}
    \def \y {2}
    \node[port] (s) at (0,0) {$s$};
    \node[port] (v1) at (\x,\y) {$v_1$};
    \node[port] (v2) at (\x, 0)  {$v_2$};
    \node[port] (v3) at (\x,-\y) {$v_3$};
    \node[port] (t) at (2*\x,0) {$t$};
    \draw[edge, draw=red!80!]            (v1) to []          node[weight, above] {1} (t);
    \draw[edge, draw=green!80!, dashed]  (v2) to []          node[weight, below] {1} (t);
    \draw[edge, draw=orange!80!, dotted] (v3) to []          node[weight, below] {1} (t);
    \node () at (\x, -\y - \offset) {$t = 1$};
\end{tikzpicture}}%
  \hspace{0.2cm}
  \scalebox{0.6}{\begin{tikzpicture}
    [port/.style={rounded corners, rectangle, fill=cyan!20!, draw=cyan!50!},
    weight/.style={align=center, midway},
    edge/.style={->, thick, -{Latex[length=3mm]}},
    scale=1]
    \def \offset {0.5};
    \def \x {2}
    \def \y {2}
    \node[port] (s) at (0,0) {$s$};
    \node[port] (v1) at (\x,\y) {$v_1$};
    \node[port] (v2) at (\x, 0)  {$v_2$};
    \node[port] (v3) at (\x,-\y) {$v_3$};
    \node[port] (t) at (2*\x,0) {$t$};
    \draw[edge, draw=blue!80!, decorate, decoration=snake ]   (s)  to []          node[weight, above] {1} (v2);
    \draw[edge, draw=blue!80!, decorate, decoration=snake ]   (v2) to []          node[weight, above] {1} (t);
    \node () at (\x, -\y - \offset) {$t = 2, 3, 4$};
\end{tikzpicture}}%
  \caption{For the \emph{single path} model, we have the path assignment in the left figure. Notice the path for green (dashed) flow shares an edge with that for the blue (curly) flow. Here is one optimal solution for the \emph{single path model}. The total weighted completion time is $1 + 1 + 1 + 4 = 7$.}
  \label{fig:singlepath}
\end{figure}
\begin{figure}[htbp]
  \centering
  \scalebox{0.6}{\begin{tikzpicture}
    [port/.style={rounded corners, rectangle, fill=cyan!20!, draw=cyan!50!},
    weight/.style={align=center, midway},
    edge/.style={->, thick, -{Latex[length=3mm]}},
    scale=1]
    \def \offset {0.5};
    \def \x {2}
    \def \y {2}
    \node[port] (s) at (0,0) {$s$};
    \node[port] (v1) at (\x,\y) {$v_1$};
    \node[port] (v2) at (\x, 0)  {$v_2$};
    \node[port] (v3) at (\x,-\y) {$v_3$};
    \node[port] (t) at (2*\x,0) {$t$};
    \draw[edge, draw=red!80!]            (v1) to []          node[weight, above] {1} (t);
    \draw[edge, draw=green!80!, dashed]  (v2) to []          node[weight, below] {1} (t);
    \draw[edge, draw=orange!80!, dotted] (v3) to []          node[weight, below] {1} (t);
    \node () at (\x, -\y - \offset) {$t = 1$};
\end{tikzpicture}}%
  \hspace{1cm}
  \scalebox{0.6}{\begin{tikzpicture}
    [port/.style={rounded corners, rectangle, fill=cyan!20!, draw=cyan!50!},
    weight/.style={align=center, midway},
    edge/.style={->, thick, -{Latex[length=3mm]}, decorate, decoration=snake },
    scale=1]
    \def \offset {0.5};
    \def \x {2}
    \def \y {2}
    \node[port] (s) at (0,0) {$s$};
    \node[port] (v1) at (\x,\y) {$v_1$};
    \node[port] (v2) at (\x, 0)  {$v_2$};
    \node[port] (v3) at (\x,-\y) {$v_3$};
    \node[port] (t) at (2*\x,0) {$t$};
    \draw[edge, draw=blue!80!] (s)  to []          node[weight, above] {1} (v1);
    \draw[edge, draw=blue!80!] (s)  to []          node[weight, below] {1} (v2);
    \draw[edge, draw=blue!80!] (s)  to []          node[weight, below] {1} (v3);
    \draw[edge, draw=blue!80!] (v1) to []          node[weight, above] {1} (t);
    \draw[edge, draw=blue!80!] (v2) to []          node[weight, below] {1} (t);
    \draw[edge, draw=blue!80!] (v3) to []          node[weight, below] {1} (t);
    \node () at (\x, -\y - \offset) {$t = 2$};
\end{tikzpicture}}%
  \caption{This is the optimal solution in the \emph{free path model}. At time $1$, send the red (solid), green (dashed), and orange (dotted) coflows. At time $2$, send the blue (curly) coflow on all paths. The total weighted completion time is $1 + 1 + 1 + 2 = 5$.}
  \label{fig:freepath}
\end{figure}
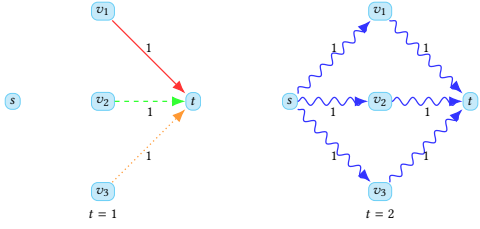


\section{Linear Programming Relaxation}
\label{sec:lp_formulation}
We use a time-indexed linear program to model this problem. Let $T$ denote an upper bound on the total time required to schedule all the coflows. Note that $T$ might be super-polynomial if the release times or coflow sizes are large. However, there is a standard technique that achieves polynomial size at the cost of a $(1 + \epsilon)$ factor on approximation ratio. We will assume $T$ to be polynomial in the main paper, and present the fix for super-polynomial $T$ in \ifjournal{Appendix~\ref{sec:large_jobs}}\else{the full version}\fi .

Let time be slotted and time slot $t$ cover the interval of time $[t-1, t]$. For a given flow $f^i_j$ and a time slot $t$, we introduce the variable $x_j^i(t)$ to indicate the fraction of flow $f^i_j$ that is scheduled at time $t$.
For each coflow $F_j$, we introduce variables $X_j(t)$ to indicate if all the flows $f^i_j \in F_j$ have been completely scheduled by time $t$. Finally, we introduce a variable $C_j$ that models the completion time of coflow $F_j$.

To make the linear program compatible with both single path model and free path model, we exclude the flow constraints and edge bandwidth constraints for now and delay them to Section~\ref{sec:flow constraints}.

\begin{align}
  \text{Minimize} \quad &\sum_jw_jC_j\nonumber\\
  \text{Subject to} \quad &\sum_tx_{j}^i(t) = 1 && \forall j \in [n], \forall i \in [n_j]\label{lp1:job finish}\\
&X_j(t) \leq \sum_{\ell = 1}^t x_j^i(\ell) && \forall j \in [n], \forall i \in [n_j], \forall t \in T \label{lp1:cumulative}\\
&C_j \geq 1 + \sum_t (1 - X_j(t)) && \forall j\in [n] \label{lp1:completion}\\
&r_j^i \geq t \Rightarrow x_{j}^i(t) = 0 && \forall j \in [n],\forall i \in [n_j], \forall t\in T\label{lp1:release time}\\
&x_{j}^i(t) \geq 0 && \forall j \in [n],\forall i \in [n_j], \forall t\in T \label{lp1:positive}
\end{align}

Constraint \eqref{lp1:job finish} certifies that each flow is fully scheduled. Constraint \eqref{lp1:cumulative} ensures that coflow $F_j$ is considered completed at time $t$ only if all flows $f_j^i \in F_j$ have been fully scheduled by time $t$. In Proposition \ref{prop:lpfeasibility}, we show that Constraint \eqref{lp1:completion} enforces a valid lower bound on the completion time of coflow $F_j$. Finally, Constraint \eqref{lp1:release time} ensures that no flow is scheduled before it has been released. Note this is not a typical LP relaxation, since any fractional solution is valid. The main relaxation is around the completion time, since representing the exact completion time of job is beyond the capability of a linear program.

\begin{proposition}
\label{prop:lpfeasibility}
The completion time of a coflow $F_j$ can be lower bounded by $C_j \geq 1 + \sum_t(1 - X_j(t))$ where $X_j(t) \in [0,1]$ denotes the fraction of coflow $F_j$ that has been completed by (the end of) time slot $t$.
\end{proposition}

\begin{proof}
Conventionally, in time-indexed linear programming relaxations, the completion time of a job $j$ is lower bounded by the fractional completion time in the schedule, or $C_j = C_j\cdot \sum_{t=1}^Tx_j(t) \geq \sum_{t=1}^Tt\cdot x_j(t)$. In our setting, this corresponds to the constraint $C_j \geq \sum_t t \cdot x_j(t)$ where $x_j(t) = X_j(t) - X_j(t-1)$ denotes the fraction of coflow $F_j$ that is scheduled during time slot $t$. The desired constraint in Eq \eqref{lp1:completion} is exactly the same constraint rearranged in a format that is more convenient for analysis.
\begin{align*}
  C_j &\geq \sum_{t = 1}^{T}t\cdot x_j(t) = \sum_{t = 1}^{T}x_j(t)\sum_{\tau=1}^t1\\
      &= \sum_{\tau=1}^{T}\sum_{t\geq \tau}^{T}x_j(t) = \sum_{\tau=1}^{T}\left(\sum_{t = 1}^{T}x_j(t) - \sum_{t = 1}^{\tau - 1}x_j(t) \right)\\
      &= \sum_{\tau=1}^{T}(1 - X_j(\tau - 1)) = \sum_{\tau=0}^{T- 1}(1 - X_j(\tau)) \\
      &= 1 + \sum_{\tau=1}^{T- 1}(1 - X_j(\tau))
\end{align*}
\end{proof}

\subsection{Model-specific Constraints}
\label{sec:flow constraints}

\subsubsection{Single Path Model}
\label{sec:path_given}

In the single path model, a flow $f^i_j$ can only be routed along a specified path $p^i_j$. Thus, we do not need to make any routing decisions in the linear program and only need to ensure that edge bandwidths are respected.
\begin{align}
  \sum_{p_j^i\ni e} x^i_j(t)\cdot \sigma^i_j &\leq c(e), && \forall e\in E, \forall t\in T\label{eq:path_not_given_constraint}
\end{align}
Constraint \eqref{eq:path_not_given_constraint} enforces that the total flow scheduled through edge $e$ at any time slot $t$ does not exceed the edge bandwidth. Constraints \eqref{lp1:job finish}-\eqref{eq:path_not_given_constraint} thus form the complete linear programming relaxation for coflow scheduling in the single path model.

\subsubsection{Free Path Model}
\label{sec:path_not_given}

In the free path model, the path for flow $f^i_j$ is not specified. In fact, data can split and merge at vertices to utilize all possible capacity. We use variable $x^i_j(t, e)$ to denote the fraction of flow $f^i_j$ transmitted through edge $e$ in time slot $t$. Recall that we use $x^i_j(t)$ to denote the total fraction of flow $f^i_j$ that is transmitted in time slot $t$. $\delta_{in}(v)$ ($\delta_{out}$) represents the set of edges that comes in (out of) vertex $v$. Here are the flow conservation constraints we need.

\begin{align}
  &\sum_{e \in \delta_{out}(s_j^i)}x^i_j(t, e) = x^i_j(t), &&\hspace{-5mm}\forall j \in [n],\forall i \in [n_j], \forall t\in T\label{eq:source}\\
  &\sum_{e \in \delta_{in}(t_j^i)}x^i_j(t, e) = x^i_j(t), &&\hspace{-5mm}\forall j \in [n],\forall i \in [n_j], \forall t\in T\label{eq:sink}\\
  &\sum_{e \in \delta_{in}(v)}x^i_j(t, e) = \sum_{e \in \delta_{out}(v)}x^i_j(t, e), &&\forall j \in [n],\forall i \in [n_j], \forall t\in T, \nonumber\\
  &&&\quad\forall v\in V\backslash \{s_j^i, t_j^i\}\label{eq:conservation}\\
  &\sum_{j \in [n],i \in [n_j]} x^i_j(t, e)\cdot \sigma^i_j \leq c(e), &&\forall t\in T, \forall e\in E\label{eq:capacity}
\end{align}

Constraints~\eqref{eq:source} and \eqref{eq:sink} enforce that the total fraction of flow $f^i_j$ satisfied at time $t$ over all the paths is exactly $x_j^i(t)$. Constraints~\eqref{eq:conservation} ensure flow conservation at all nodes other than source and sink. Constraints~\eqref{eq:capacity} guarantee that all edge bandwidths are satisfied at all time steps.
Constraints \eqref{lp1:job finish}-\eqref{lp1:positive} and \eqref{eq:source}-\eqref{eq:capacity} thus form the complete linear programming relaxation for coflow scheduling in the free path model.

Let $C_j^*$ denote the completion time of coflow $F_j$ in an optimal solution of the LP relaxation, and let $C_j(opt)$ denote the completion time of coflow $F_j$ in the corresponding optimal integral solution. Thus, for both the models, we have
\begin{align}
\sum_j w_j C_j^* &\leq \sum_j w_j C_j(opt). \label{eq:lp_opt}
\end{align}


\section{Approximation Algorithms}

Let $x_j^i(t)$ denote the fraction of flow $f_j^i$ that is scheduled at time step $t$ in an optimal solution to the above LP. The LP constraints guarantee that this yields a feasible schedule to the coflow scheduling problem (in both the single path as well as the free path models). However, since the completion time of a coflow $F_j$ is defined as the earliest time $t$ such that all flows $f_j^i \in F_j$ have been completely scheduled, the true completion time of coflow $F_j$ obtained in this scheduled is given by
\begin{align}
C_j(\text{\it LP Sched}) = \max_i \{\max_{t : x_j^i(t) > 0} [t]\}.
\end{align}
Unfortunately, this completion time $C_j(\text{\it LP Sched})$ can be much greater than the completion time variable in the optimal LP solution $C_j^*$, and thus the obtained schedule is not a constant-approximate coflow schedule. For instance, consider a coflow $F_j$ with only one flow ($n_j = 1$) and let the optimal LP solution set its schedule as follows $x^1_j(1) = 0.9, x^1_j(10) = 0.1$, and $x^1_j(t) = 0, \forall t \notin \{1,10\}$. Now, the completion time variable in the optimal LP solution is $C_j^* = \sum_t t x_j^1(t) = 1.9$. However, true completion time of the  coflow $F_j$ in such a schedule is $C_j(\text{\it LP Sched}) = 10 \gg C_j^*$.

To overcome the obstacle above, we propose the following algorithm called \textsf{Stretch} (see Section~\ref{sec:algorithm}) that modifies the schedule obtained by the linear program so that the completion time of each coflow in the modified schedule can be compared with the completion time variable of the corresponding coflow in an optimal LP solution. The schedule ``stretching'' idea (also called `slow-motion') used in our algorithm has been used before successfully in other scheduling contexts~\cite{im2014preemptive,queyranne20022+,schulz1997random}.

\subsection{\textsf{Stretch} Algorithm}
\label{sec:algorithm}
\begin{enumerate}
  \item Solve the linear program in Section~\ref{sec:lp_formulation} and obtain a fractional optimal solution.
  \item Let $\lambda \in (0,1)$ be drawn randomly according to the p.d.f $f(v) = 2v$. We can verify that this is indeed a valid probability distribution.
  \item Stretch the LP schedule by $\frac{1}{\lambda}$. This means that we schedule everything exactly as per the LP solution - but whatever LP schedules in the interval $[a,b]$, we will schedule in the interval $[\frac{a}{\lambda}, \frac{b}{\lambda}]$.
  \item Once $\sigma_j^i$ units of flow $f_j^i$ have been scheduled, leave the remaining slots for $f_j^i$ empty.
\end{enumerate}

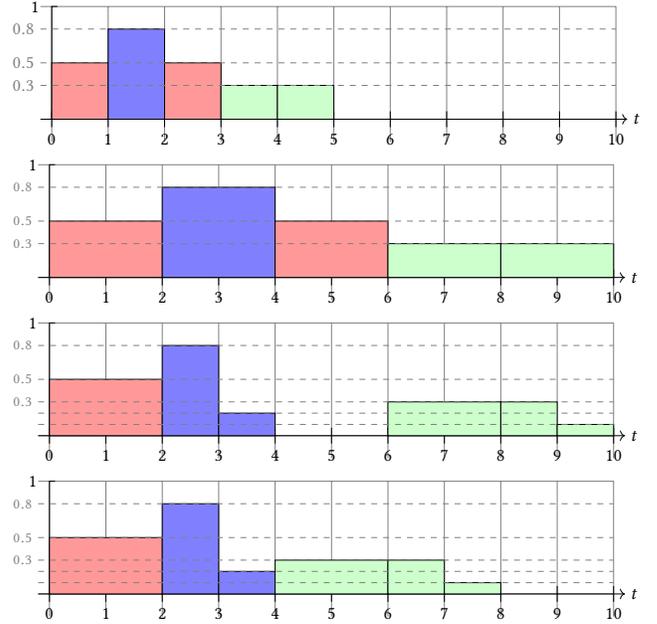
\begin{figure}[htbp]
  \centering
  \scalebox{0.75}{\begin{tikzpicture}
  [
    job1/.style={fill=red!40!},
    job2/.style={fill=blue!50!},
    job3/.style={fill=green!20!},
    yscale=2,
  ]
  \def \offset {0.2}
  \def \xAxis {10}
  \draw[gray, very thin] (-\offset, -\offset) grid ++(\xAxis + \offset, 1 + \offset);
  \draw[->] (0 - \offset,0) -- (\xAxis + \offset,0) node[right] {$t$};
  \foreach \pos in {0,1,...,\xAxis}
    \draw[shift={(\pos,0)}] (0pt,2pt) -- (0pt,-2pt) node[below] {$\pos$};
  \draw (0,0) -- ++(0, 1) -- ++(0.1, 0);
  \node[label={left:{1}}] (label_1) at (0, 1) {};
  \draw[job1] (0,0) rectangle ++(1, 0.5);
  \draw[job2] (1,0) rectangle ++(1, 0.8);
  \draw[job1] (2,0) rectangle ++(1, 0.5);
  \draw[job3] (3,0) rectangle ++(1, 0.3);
  \draw[job3] (4,0) rectangle ++(1, 0.3);
  \foreach \value in {0.3, 0.5, 0.8}
    \draw[gray, very thin, dashed] (0 - \offset,\value) node[left] {$\value$} -- (\xAxis,\value);
\end{tikzpicture}}
  \scalebox{0.75}{\begin{tikzpicture}
  [
    job1/.style={fill=red!40!},
    job2/.style={fill=blue!50!},
    job3/.style={fill=green!20!},
    yscale=2,
  ]
  \def \offset {0.2}
  \def \xAxis {10}
  \draw[gray, very thin] (-\offset, -\offset) grid ++(\xAxis + \offset, 1 + \offset);
  \draw[->] (0 - \offset,0) -- (\xAxis + \offset,0) node[right] {$t$};
  \foreach \pos in {0,1,...,\xAxis}
    \draw[shift={(\pos,0)}] (0pt,2pt) -- (0pt,-2pt) node[below] {$\pos$};
  \draw (0,0) -- ++(0, 1) -- ++(0.1, 0);
  \node[label={left:{1}}] (label_1) at (0, 1) {};
  \draw[job1] (0,0) rectangle ++(2, 0.5);
  \draw[job2] (2,0) rectangle ++(2, 0.8);
  \draw[job1] (4,0) rectangle ++(2, 0.5);
  \draw[job3] (6,0) rectangle ++(2, 0.3);
  \draw[job3] (8,0) rectangle ++(2, 0.3);
  \foreach \value in {0.3, 0.5, 0.8}
    \draw[gray, very thin, dashed] (0 - \offset,\value) node[left] {\footnotesize{$\value$}} -- (\xAxis,\value);
\end{tikzpicture}}
  \scalebox{0.75}{\begin{tikzpicture}
  [
    job1/.style={fill=red!40!},
    job2/.style={fill=blue!50!},
    job3/.style={fill=green!20!},
    yscale=2,
  ]
  \def \offset {0.2}
  \def \xAxis {10}
  \draw[gray, very thin] (-\offset, -\offset) grid ++(\xAxis + \offset, 1 + \offset);
  \draw[->] (0 - \offset,0) -- (\xAxis + \offset,0) node[right] {$t$};
  \foreach \pos in {0,1,...,\xAxis}
    \draw[shift={(\pos,0)}] (0pt,2pt) -- (0pt,-2pt) node[below] {$\pos$};
  \draw (0,0) -- ++(0, 1) -- ++(0.1, 0);
  \node[label={left:{1}}] (label_1) at (0, 1) {};
  \draw[job1] (0,0) rectangle ++(2, 0.5);
  \draw[job2] (2,0) rectangle ++(1, 0.8);
  \draw[job2] (3,0) rectangle ++(1, 0.2);
  \draw[job3] (6,0) rectangle ++(2, 0.3);
  \draw[job3] (8,0) rectangle ++(1, 0.3);
  \draw[job3] (9,0) rectangle ++(1, 0.1);
  \foreach \value in {0.3, 0.5, 0.8}
    \draw[gray, very thin, dashed] (0 - \offset,\value) node[left] {\footnotesize{$\value$}} -- (\xAxis,\value);
  \foreach \value in {0.1, 0.2}
    \draw[gray, very thin, dashed] (0 - \offset,\value) -- (\xAxis,\value);
\end{tikzpicture}}
  \scalebox{0.75}{\begin{tikzpicture}
  [
    job1/.style={fill=red!40!},
    job2/.style={fill=blue!50!},
    job3/.style={fill=green!20!},
    yscale=2,
  ]
  \def \offset {0.2}
  \def \xAxis {10}
  \draw[gray, very thin] (-\offset, -\offset) grid ++(\xAxis + \offset, 1 + \offset);
  \draw[->] (0 - \offset,0) -- (\xAxis + \offset,0) node[right] {$t$};
  \foreach \pos in {0,1,...,\xAxis}
    \draw[shift={(\pos,0)}] (0pt,2pt) -- (0pt,-2pt) node[below] {$\pos$};
  \draw (0,0) -- ++(0, 1) -- ++(0.1, 0);
  \node[label={left:{1}}] (label_1) at (0, 1) {};
  \draw[job1] (0,0) rectangle ++(2, 0.5);
  \draw[job2] (2,0) rectangle ++(1, 0.8);
  \draw[job2] (3,0) rectangle ++(1, 0.2);
  \draw[job3] (4,0) rectangle ++(2, 0.3);
  \draw[job3] (6,0) rectangle ++(1, 0.3);
  \draw[job3] (7,0) rectangle ++(1, 0.1);
  \foreach \value in {0.3, 0.5, 0.8}
    \draw[gray, very thin, dashed] (0 - \offset,\value) node[left] {\footnotesize{$\value$}} -- (\xAxis,\value);
  \foreach \value in {0.1, 0.2}
    \draw[gray, very thin, dashed] (0 - \offset,\value) -- (\xAxis,\value);
\end{tikzpicture}}
  \caption{Here we show an example solution obtained from the LP, different color indicate different flows. In the second picture, we stretch with $\lambda = 0.5$. In the third picture, we leave the slots empty if the corresponding flow is finished. In the fourth picture, we utilize the idle slots and move some flows to earlier times. Though this does not improve the theoretically bound, it is beneficial in practice and is used in our experimental evaluation.}
  \label{pic:stretching}
\end{figure}

Figure \ref{pic:stretching} illustrates the key ideas of the algorithm.
To help understand this algorithm, start with the simple case where we have a fixed $\lambda = 0.5$, in other words stretch the time axis by a factor of $1/\lambda = 2$. Intuitively, we move everything at time slot $t$ and to both time slots $2t-1$ and $2t$. What used to be transmitted at time $t$ will be transmitted no later than time $2t$. Consider any flow $f^i_j$ and let $\tau$ denote the earliest time by which the LP has scheduled at least $1/2$ fraction of the flow. Then, it is easy to verify that the flow $f^i_j$ is completely scheduled by time $2\tau$.

Now we consider a general $\lambda$ and prove that this algorithm does output a feasible schedule. Due to fractional $\lambda$, it might be the case that some flow $f_j^i$ of LP variable $x_j^i(t)$ in integral interval $[t - 1, t]$ becomes $[\frac{t - 1}{\lambda}, \frac{t}{\lambda}]$, a fractional interval. In this case, for a time slot $\tau$, or a interval $[\tau - 1, \tau]$ after stretching, we just add $x_j^i(t) \cdot |[\tau - 1, \tau]\cap [\frac{t - 1}{\lambda}, \frac{t}{\lambda}]|$.

The only flows that might be scheduled in time slot $\tau$ are those scheduled in time slot $1 + \lfloor \lambda(\tau - 1) \rfloor$ and $1 + \lfloor \lambda \tau \rfloor$ before stretching, or flows $f_j^i(1 + \lfloor \lambda(\tau - 1) \rfloor)$ and flows $f_j^i(1 + \lfloor \lambda \tau \rfloor)$. (The two time slots might be the same. If so, feasibility is automatically met. Otherwise, we have $1 + \lfloor \lambda(\tau - 1) \rfloor + 1 = 1 + \lfloor \lambda \tau \rfloor$.) For all flows at time $1 + \lfloor \lambda(\tau - 1) \rfloor$ before stretching, the factor we multiplied with is $w_1 = \left| [\tau - 1, \tau] \cap [\frac{\lfloor \lambda(\tau - 1) \rfloor}{\lambda}, \frac{1 + \lfloor \lambda(\tau - 1) \rfloor}{\lambda}]\right|$. For all flows at time $ 1 + \lfloor \lambda \tau \rfloor$ before stretching, the factor we use to multiply with is $w_2 = \left| [\tau - 1, \tau] \cap [\frac{\lfloor \lambda \tau \rfloor}{\lambda}, \frac{1 + \lfloor \lambda \tau\rfloor}{\lambda}]\right|$. Note $w_1 +w_2 = 1$. In fact, the schedule at time $\tau$ can be viewed as a weighted average of the schedule at time $[\lfloor \lambda (\tau - 1) \rfloor, 1 + \lfloor \lambda(\tau - 1) \rfloor]$ and $[\lfloor \lambda \tau \rfloor, 1 + \lfloor \lambda \tau \rfloor]$ (if $\lambda(\tau - 1)$ is a integer, then the schedule will be exactly what it used to be at time $\lambda\tau$), the first with weight $w_1$ and the second with weight $w_2$. The nature of network flow ensures that the weighted sum of two feasible flows is a feasible flow.

Another fact that needs proof is that every flow is finished. This is guaranteed since schedules are stretched, and we only leave the remaining slots empty for $f_j^i$ if $\sigma_j^i$ units of flow have been scheduled, or in other words, all the demand for this flow has been scheduled.

\subsection{Analysis}
\label{sec:analysis-random}

Recall that $C_j^*$ denotes the completion time of coflow $F_j$ in the optimal LP solution. While we consider that time is slotted in the LP formulation and time slot $t$ covers the interval of time $[t-1, t]$, at this stage it is more convenient to work with continuous time rather than discrete time. 
For any continuous time $\tau \in [0,T]$, define $X_j(\tau)$ to be the fraction of coflow $F_j$ that has been scheduled in the LP solution by time $\tau$. We define $X_j(\tau)$ by assuming that the flow is scheduled at an uniform rate in every time slot. Formally, we have
\begin{align}
X_j(\tau) = X_j(\lfloor \tau \rfloor) + (\tau - \lfloor \tau \rfloor) \left( X_j(\lfloor \tau \rfloor + 1) - X_j(\lfloor \tau \rfloor) \right).
 \label{eq:tau}
\end{align}

The LP constraints \eqref{lp1:completion} guarantee that for any coflow $F_j$, we have $C_j^* \geq 1 + \sum_t(1 - X_j(t))$.
We can now lower-bound the LP completion time by replacing the above summation by an integral. 

\begin{lemma}
\label{lem:integral}
$\int_{\tau = 0}^T (1 - X_j(\tau)) d\tau \leq C_j^* - \frac{1}{2}$ where $X_j(\tau)$ is defined as per Eq. \eqref{eq:tau}.
\end{lemma}

\begin{proof}
By definition of $X_j(\tau)$, we have the following.
\begin{align*}
    \int_{\tau = 0}^T (1 - &X_j(\tau)) d\tau = T - \int_{\tau=0}^T X_j(\tau) d\tau\\
    &= T - \sum_{t=0}^{T-1} \int_{\tau=t}^{t+1} X_j(\tau) d\tau\\
    &= T - \sum_{t=0}^{T-1} \int_{\tau=t}^{t+1} \left[X_j(t) + (\tau - t) \left( X_j(t + 1) - X_j(t) \right)\right] d\tau\\
    &= T - \sum_{t=0}^{T-1} \left[X_j(t) + \left( X_j(t + 1) - X_j(t) \right) \int_{\tau=t}^{t+1}  (\tau - t) d\tau \right]\\
    &= T - \sum_{t=0}^{T-1} \frac{1}{2}\left[X_j(t) + X_j(t + 1) \right]\\
    &= T - \left[\frac{1}{2} \left(X_j(0) + X_j(T)\right) +  \sum_{t=1}^{T-1} X_j(t) \right]
\intertext{Since by definition, $X_j(0) = 0$ and $X_j(T) = 1$, we get}
    &= T - \left[\frac{1}{2} +  \sum_{t=1}^{T-1} X_j(t) \right]
\intertext{Rearranging the terms, we get}
    &= 1 + \sum_{t=1}^{T-1}(1 - X_j(t)) - \frac{1}{2}\\
    &\leq C_j^* - \frac{1}{2}
\end{align*}
where the last inequality follows from Constraint \eqref{lp1:completion}.
\end{proof}

For any $\lambda \in [0,1]$, define $C_j^*(\lambda)$ to be the earliest time $\tau$ such that $\lambda$ fraction of the coflow $F_j$ has been scheduled in the LP solution, i.e., in other words its the smallest $\tau$ such that $X_j(\tau) = \lambda$. Note that by time $C_j^*(\lambda)$, $\lambda$ fraction of \emph{every flow $f^i_j \in F_j$} has been scheduled by the LP.

\begin{proposition}
\label{prop:counting}
$ \displaystyle \int_{\lambda = 0}^1 C^*_j(\lambda) \mathbf{d}\lambda = \int_{\tau = 0}^T (1 - X_j(\tau)) \mathbf{d}\tau$
\end{proposition}

\begin{proof}
\begin{align*}
\int_{\lambda = 0}^1 C^*_j(\lambda) \mathbf{d}\lambda &= \int_{\lambda = 0}^1 \int_{\tau = 0}^T \mathbbm{1}_{[C^*_j(\lambda) > \tau]} \mathbf{d}\tau \mathbf{d}\lambda\\
&= \int_{\tau = 0}^T \int_{\lambda = 0}^1 \mathbbm{1}_{[C^*_j(\lambda) > \tau]} \mathbf{d}\lambda \mathbf{d}\tau\\
&= \int_{\tau = 0}^T \int_{\lambda = X_j(\tau)}^1 1 \mathbf{d}\lambda \mathbf{d}\tau\\
&= \int_{\tau = 0}^T (1 - X_j(\tau)) \mathbf{d}\tau
\end{align*}
\end{proof}

Finally, we are ready to bound the completion time of coflow $F_j$ in the stretched schedule (denoted as $C_j(alg)$). For any fixed $\lambda \in (0,1)$, since we stretch the schedule by a factor of $\frac{1}{\lambda}$, we have $C_j(alg) \leq \left\lceil \frac{C_j^*(\lambda)}{\lambda}\right\rceil$. Notice the ceiling function in the bound~\footnote{All flows $f_j^i \in F_j$ were completed by at least $\lambda$ fraction by time $C^*_j(\lambda)$. So in the stretched schedule, all those flows must be completed by  time $\frac{C^*_j(\lambda)}{\lambda}$. The ceiling is necessary since $\frac{C^*_j(\lambda)}{\lambda}$ may be fractional (i.e. occur in the middle of a time slot)}. Since $\lambda$ is drawn randomly from a distribution, the following lemma bounds the expected completion time of coflow $F_j$ in the stretched schedule.

\begin{lemma}
The expected completion time of any coflow $F_j$ in the stretched schedule is bounded by $2C_j^*$.
\end{lemma}

\begin{proof}
\begin{align*}
\mathbb{E}[C_j(alg)] &\leq \int_{\lambda=0}^1 f(\lambda) \left\lceil \frac{C^*_j(\lambda)}{\lambda}\right\rceil d\lambda\\
&\leq\int_{\lambda=0}^1 \left(2 \lambda\right) \left(\frac{C^*_j(\lambda)}{\lambda} + 1 \right) d\lambda\\
&= 2 \int_{\lambda=0}^1 C^*_j(\lambda) d\lambda + 1
\intertext{By Lemma \ref{lem:integral} and Proposition \ref{prop:counting},}
&= 2 \int_{\tau = 0}^T (1 - X_j(\tau)) \mathbf{d}\tau + 1\\
&\leq 2 \left(C_j^* - \frac{1}{2}\right) + 1 = 2 C_j^*
\end{align*}
\end{proof}

Theorem \ref{thm:main} thus follows from the linearity of expectation.

\begin{theorem}
\label{thm:main} There is a randomized 2-approximation algorithm for coflow scheduling  in networks in both the single path and free path models when all release times and coflow sizes are polynomially sized.
\end{theorem}

For the case where the total time we need to schedule all coflows is super-polynomial, we use the standard trick of geometric series time intervals, and claim the following theorem. Proof comes in \ifjournal{Appendix~\ref{sec:large_jobs}}\else{the full version}\fi .

\begin{theorem}
  \label{thm:main_super_polynomial}
  For any $\epsilon > 0$, there is a randomized $(2 + \epsilon)$-approximation algorithm for coflow scheduling  in networks in both the single path and the free path models (with possibly super polynomial release times and demands).
\end{theorem}

\section{Hardness of Approximation}
\label{sec:hardness}
\newcommand{\cC}{\mathbb{C}}
We claim the following theorem:

\begin{theorem}
  \label{thm:hardness}
  For the coflow scheduling problem, in both the single path and the free path model, it is NP-hard to obtain a $(2-\epsilon)$ approximation, for any $\epsilon > 0$.
\end{theorem}

\begin{proof}
We prove it by a reduction from concurrent open-shop problem (proved NP-hard to approximate within a factor better than $(2-\epsilon)$~\cite{bansal2010inapproximability,sachdeva2013optimal}). The definition of concurrent open shop problem is as follows: there are $m$ machines and $n$ jobs, each job $j$ need to be processed on machine $i$ for $p_j^i$ time non-preemptively. We would like to minimize the total weighted completion time. Unlike the open shop problem, in the concurrent open shop problem a job can be processed on more than one machine at the same time.

Given a concurrent open-shop problem instance with $M$ machines, we construct an instance of the coflow scheduling problem as follows. For every machine $i$, we have two nodes $x_i$ and $y_i$, and an edge of unit bandwidth from $x_i$ to $y_i$. Notice the graph has $M$ different components, between each pair $(x_i, y_i)$, there is only one path from $x_i$ to $y_i$. Thus this construction works for both the single path model and the free path model. We will not distinguish the models in the following proof.

For a certain job $j$ with demands $\sigma_j^i$ in the concurrent open shop instance, we add a coflow $j$ with demand of $\sigma_j^i$ from $x_i$ to $y_i$. Weights are directly taken from the concurrent open shop problem instance. Suppose we get a solution for this coflow scheduling instance, we can get a solution of no larger cost for the concurrent open shop instance as follows. If we have a flow $f_j^i$ for job $j$ on edge $(x_i, y_i)$ of size $x_j^i(t)$ at time $t$, then we schedule a fraction of $x_j^i(t)$ for job $j$ on machine $i$ at time $t$. Suppose a flow $f_j^i$ is finished at time $C_j^i$ in the coflow scheduling problem, the corresponding concurrent open shop problem for job $j$ and machine $i$ is also finished at time $C_j^i$. Similarly, the finishing time $C_j$ of coflow $j$ and concurrent open shop job $j$ are the same. However, the solution we get is fractional, and might be preemptive (we might pause a job and resume it later).

Now we prove that we can modify this solution to get a non-preemptive integral solution without raising the total weighted completion time. For each machine $i$, consider all completion times $C_j^i$. Sort them in non-decreasing order $C_{l_1}^i, C_{l_2}^i, \dots, C_{l_J}^i$, and we can safely reschedule these demand in the order of $l_1, l_2, \dots, l_j$, and get new completion times $\cC_{l_1}^i, \dots, \cC_{l_j}^i$ while not raising any completion time. We know all demand of job $l_1$ on machine $i$ has been finished by $C_{l_1}^i$, so $\cC_{l_1}^i = d_{l_1, i} \leq C_{l_1}^i$, similarly all demands of job $l_1$ and $l_2$ have been finished by $C_{l_2}^i$, and $\cC_{l_2}^i = d_{l_1, i} + d_{l_2, i} \leq C_{l_2}^i$. We can continue and get $\cC_{l_j}^i \leq C_{l_j}^i, \forall j\in \{J\}$. Thus the total weighted completion time for this integral solution would be upper bounded by the cost for the coflow scheduling instance.
\begin{align*}
\sum_{j\in \{J\}}w_j\cdot \cC_j &= \sum_{j\in \{J\}}w_j\cdot \max_{i\in \{M\}}\cC_j^i\\
                                 &\leq \sum_{j\in \{J\}}w_j\cdot \max_{i\in \{M\}}C_j^i\\
                                 &= \sum_{j\in \{J\}}w_j\cdot C_j
\end{align*}

For the other direction, for a certain solution of a concurrent open-shop problem, if task $i$ of job $j$ is scheduled from time $t_1$ to time $t_2$, we make the flow $f_j^i$ take up all bandwidth of edge $(x_i, y_i)$ from time $t_1$ to time $t_2$. Then flow $f_j^i$ is finished the same time when task $i$ of job $j$ is finished. Since every task $i$ is finished the same time before and after reduction, completion times and the objective weighted completion time stays the same for the coflow scheduling problem.

In conclusion, for a solution $SOL$ of concurrent open-shop problem with weighted completion time $W$, we can construct a solution $SOL_{\text{coflow}}$ for coflow scheduling problem of the same weighted completion time $W$. For a solution $SOL'_{\text{coflow}}$ of coflow scheduling problem with weighted completion time $W'$, we can construct a solution $SOL'$ for the original concurrent open-shop problem, with cost at most $W'$. Since concurrent open-shop problem is NP-hard to get a $(2 - \epsilon)$ approximation, we know it is also NP-hard to approximate coflow scheduling problem to a factor of $(2 - \epsilon)$, for both single path model and free path model.
\end{proof}

\section{Experiments}
\label{sec:experiments}
We evaluated the \textsf{Stretch} Algorithm on 2 topologies and 4 benchmarks/industrial workloads. Experiments were run on a machine with dual \textsf{Intel(R) Xeon(R) CPU E5-2430}, and 64GB of RAM, and using \textsf{Gurobi}~\cite{gurobi} as the LP solver. We first discuss the experimental set up and then in Section~\ref{sec:baseline} discuss what evaluation we performed.

\textbf{WAN topology:} We consider the following graph topologies.
\begin{enumerate}
\item Swan~\cite{hong2013achieving}: Microsoft's inter-datacenter WAN with 5 datacenters and 7 inter-datacenter links. We calculate link bandwidth using the setup described by Hong et al.\cite{hong2013achieving}.
\item G-Scale~\cite{jain2013b4}: Google's inter-datacenter WAN with 12 datacenters and 19 inter-datacenter links.
\end{enumerate}

\textbf{Workloads:} We use the following mix of jobs from public benchmarks - TPC-DS~\cite{nambiar2006making}, TPC-H~\cite{poess2000new}, and BigBench~\cite{big-bench} - and from Facebook (FB) production traces~\cite{SWIM,coflow-benchmark}. We follow \cite{chowdhury2018terra} to set up the benchmarks: for a certain workload, jobs are randomly chosen and since they do not have a release time, we assign a release time similar to that in production traces. Each job lasts from a few minutes to dozens of minutes. Each benchmark experiment has 200 jobs. We randomly assign these jobs to nodes in the datacenter, and the demand will be between the corresponding nodes. Since weights are not available, we assign weights that are uniformly chosen from the interval between $1.0$ and $100.0$.

\subsection{Implementation Details}
\label{sec:impl_details}
In this subsection we discuss some details related to the implementation.

\noindent
\textbf{Time Index:} There is a trade-off in selecting the size of a time slot.
If the length of a time slot is shorter, we get more accurate answers, but need to solve a 
larger LP. On the contrary, if we make each time slot longer, the amount of computational resources need is greatly reduced, but the quality of the solution suffers. In all our experiments, we considered time slots of length $50$ seconds as this led to tractable LP relaxations.

\noindent
\textbf{Rounding:} Algorithm \textsf{Stretch} is meant for easy theoretical analysis, and is not a sophisticated rounding
method; we are not trying to schedule later flows in the slots that are idle. This can cause huge overhead in experiments. See Figure~\ref{pic:stretching} for an illustration. In our implementation, we  deal with this issue by moving the schedule of every time slot $t$ to an earlier idle slot $t'$ if for all flows scheduled at $t$, its release time is before $t'$. 

To address the random sampling of $\lambda$, we sample $20$ times from the distribution mentioned in Section~\ref{sec:algorithm} to get the expected weighted completion time for Algorithm Stretch, and denote it with ``Average $\lambda$''. We also measure the best solution obtained over these random choices (denoted by ``Best $\lambda$'').

\subsection{Baselines}
\label{sec:baseline}
\textbf{LP-based Heuristic:} In addition to algorithms with theoretically worst case guarantee, we also propose a heuristic that works well in practice. Recall in Section~\ref{sec:algorithm}, we mentioned that the LP solution itself is a valid schedule. We can use this solution as a heuristic, for both the single path and free path models. Note the weighted completion time for this LP solution is \emph{NOT} the same as the LP objective function, as explained in Section~\ref{sec:algorithm}. This implies that the solution from the heuristic can be arbitrarily bad in the worst case. In practice, however, this proves to be a very effective algorithm that can be quite close to the lower bound we get from LP.

\textbf{Jahanjou et al.\ (\emph{Single path model}):} Since path information is not available in the datasets, we randomly generate one for each flow. For a source sink pair $(s_j^i, t_j^i)$, we randomly select one of the shortest paths as the path for flow $f_j^i$. For this model, we compare our algorithm with the algorithm presented by Jahanjou et al.~\cite{jahanjou2017asymptotically}. Here is a brief description of their approach. First write an LP using geometric time intervals, then schedule each job according to the interval its $\alpha$ point (the time when $\alpha$ fraction of this job is finished) belongs to. A common reason for geometric time intervals is to avoid having a super-polynomial time horizon (a practical reason is to make the LP smaller), and a time series of $\{(1 + \epsilon)^i\}$ is chosen where $\epsilon$ is close to 0. The closer $\epsilon$ is to 0, the better the approximation ratio can be. However, in Jahanjou et al.'s algorithm, the rounding step has a dependency on $\epsilon$. To optimize the approximation ratio, $\epsilon$ is set to $0.5436$. Our algorithm, on the contrary, is time slot based, and can be turned into a geometric series of time intervals by losing a factor of $(1 + \epsilon)$. In experiments, we include both the case of $\epsilon = 0.2$ and the case of $\epsilon = 0.5436$ for completeness.

\textbf{Terra (\emph{Free path model}):} For the flow-based model, we are comparing to the offline algorithm in Terra~\cite{chowdhury2018terra}. This algorithm only works for the unweighted case. It calculates the time for each single coflow to finish individually, and then schedule with SRTF (shortest remaining time first). Instead of one large LP like all other algorithms compared here, this algorithm solves a large number of LPs, twice the number of coflow jobs. Terra can work with very fine grained time, to the order of milliseconds (and does not need time to be slotted). Since there is no previous work on weighted case, we compare the weighted case with the LP solution and our heuristic directly based on time indexed LP.

\subsection{Experimental Results}
\label{sec:experimental_results}
\begin{figure*}
  \centering
  \scalebox{0.8}{%
    \begin{tikzpicture}
 \centering
  \begin{axis}[
    ybar,
    enlargelimits=true,
    axis on top,
    title={Solution for SWAN (less is better)},
    height=8cm, width=20cm,
    bar width=12,
    ymajorgrids, tick align=inside,
    major grid style={draw=white},
    ymin=0,
    axis x line*=bottom,
    axis y line*=left,
    y axis line style={opacity=50},
    enlarge x limits=0.15,
    legend style={
      at={(0.5,-0.2)},
      anchor=north,
      legend columns=-1,
      /tikz/every even column/.append style={column sep=15}
    },
    ylabel={Weighted Completion Time},
    xtick=data,
    symbolic x coords={
      BigBench,
      TPC-DS,
      TPC-H,
      FB},
    ]
    \addplot [draw=none, fill=blue!60] coordinates {
      (BigBench, 4940095.478798)
      (TPC-DS, 4759818.526494) 
      (TPC-H, 5437756.440628)
      (FB, 5475312.274120) };
   \addplot [draw=none,fill=red!60] coordinates {
      (BigBench, 5420742.084466)
      (TPC-DS, 5229347.588657) 
      (TPC-H, 5970175.003998)
      (FB, 5998893.898198) };
   \addplot [draw=none, fill=black!60] coordinates {
      (BigBench, 5519347.811754)
      (TPC-DS, 5264346.637561) 
      (TPC-H, 5980962.920550)
      (FB, 6106025.294892) };
   \addplot [draw=none, fill=black!30] coordinates {
     (BigBench, 5519347.811754)
      (TPC-DS, 5321352.030816) 
      (TPC-H, 6072112.957983)
      (FB, 6154595.016017) };
    \legend{LP(lower bound), Heuristic($\lambda = 1.0$),%
      Best $\lambda$, Average $\lambda$}
  \end{axis}
\end{tikzpicture}
  }
  \caption{Free path model on SWAN, showing the performance bound of time indexed LP value, the performance of heuristic ($\lambda = 1$), best $\lambda$ among samples, and the expected value when $\lambda$ is chosen from the distribution mentioned in Section~\ref{sec:algorithm}.}
  \label{fig:experiment_flow_swan}
\end{figure*}
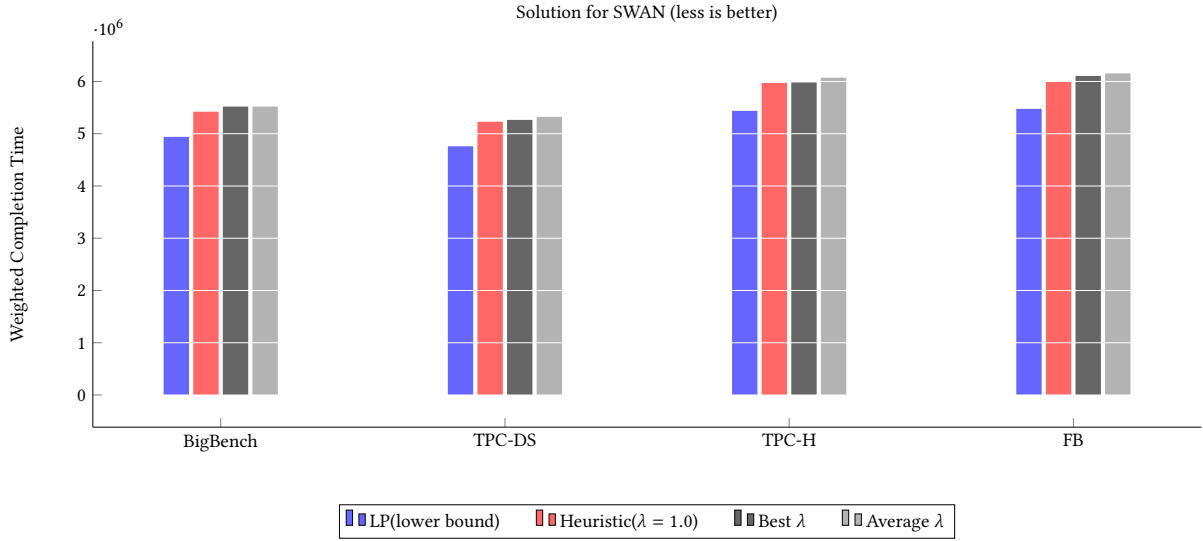
\begin{figure*}
  \centering
  \scalebox{0.8}{%
    \begin{tikzpicture}
 \centering
  \begin{axis}[
    ybar,
    enlargelimits=true,
    axis on top,
    title={Solution for G-Scale (less is better)},
    height=8cm, width=20cm,
    bar width=12,
    ymajorgrids, tick align=inside,
    major grid style={draw=white},
    ymin=0,
    axis x line*=bottom,
    axis y line*=left,
    y axis line style={opacity=50},
    enlarge x limits=0.15,
    legend style={
      at={(0.5,-0.2)},
      anchor=north,
      legend columns=-1,
      /tikz/every even column/.append style={column sep=15}
    },
    ylabel={Weighted Completion Time},
    xtick=data,
    symbolic x coords={
      BigBench,
      TPC-DS,
      TPC-H,
      FB},
    ]
    \addplot [draw=none, fill=blue!60] coordinates {
      (BigBench, 5242688.573051)
      (TPC-DS, 5202072.170825) 
      (TPC-H, 4707164.815620)
      (FB, 5242688.573051) };
   \addplot [draw=none,fill=red!60] coordinates {
      (BigBench, 5764974.093929)
      (TPC-DS, 5710642.526254) 
      (TPC-H, 5171114.886991)
      (FB, 5764974.093929) };
   \addplot [draw=none, fill=black!60] coordinates {
      (BigBench, 5776393.403719)
      (TPC-DS, 5814476.389934) 
      (TPC-H, 5265068.815928)
      (FB, 5776393.403719) };
   \addplot [draw=none, fill=black!30] coordinates {
     (BigBench, 5864953.623431)
      (TPC-DS, 5814476.389934) 
      (TPC-H, 5265068.815928)
      (FB, 5864953.623431) };
    \legend{LP(lower bound), Heuristic($\lambda = 1.0$),%
      Best $\lambda$, Average $\lambda$}
  \end{axis}
\end{tikzpicture}
  }
  \caption{Free path model on G-Scale, showing the performance bound of time indexed LP value, the performance of heuristic ($\lambda = 1$), best $\lambda$ among samples, and the expected value when $\lambda$ is chosen from the distribution mentioned in Section~\ref{sec:algorithm}.}
  \label{fig:experiment_flow_gb4}
\end{figure*}
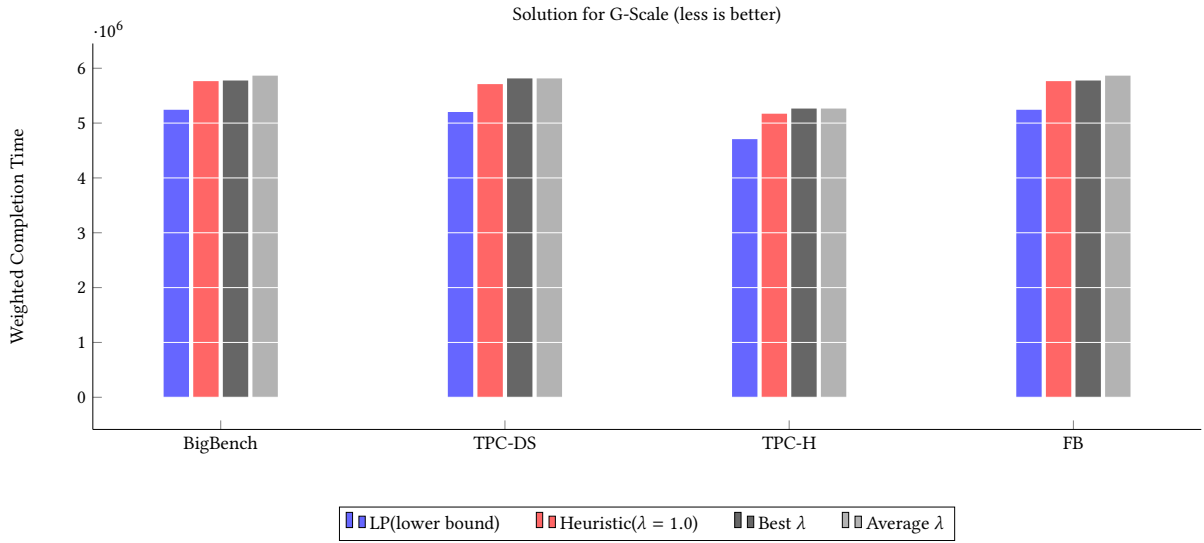
\begin{figure*}
  \centering
  \scalebox{0.7}{%
    \begin{tikzpicture}
 \centering
  \begin{axis}[
    ybar,
    enlargelimits=true,
    axis on top,
    title={Different Choice of $\epsilon$},
    height=8cm, width=20cm,
    bar width=12,
    ymajorgrids, tick align=inside,
    major grid style={draw=white},
    ymin=0,
    axis x line*=bottom,
    axis y line*=left,
    y axis line style={opacity=50},
    enlarge x limits=0.15,
    ylabel={Weighted Completion Time},
    xtick=data,
    legend style={
      at={(0.5,-0.2)},
      anchor=north,
      legend columns=-1,
      /tikz/every even column/.append style={column sep=15}
    },
    ]

   \addplot [draw=none, fill=blue!80] coordinates {
     (0.1, 5474895.078690)
     (0.2, 5801102.090175)
     (0.3, 5864595.792970)
     (0.4, 6367494.428345)
     (0.5, 6450640.357699)
     (0.6, 6529133.664135)
     (0.7, 7098737.640593)
     (0.8, 7655049.631958)
     (0.9, 8081376.374533)
     (1.0, 7361082.785002)
    };
   \addplot [draw=none, fill=red!80] coordinates {
     (0.1, 6015298.604863)
     (0.2, 6953022.289864)
     (0.3, 7616093.644334)
     (0.4, 8904294.526330)
     (0.5, 9663258.601444)
     (0.6, 10432832.074131)
     (0.7, 12053684.129923)
     (0.8, 13765720.798567)
     (0.9, 15339950.243750)
     (1.0, 14701344.177684)
    };
    \legend{%
      Time interval LP(lower bound), heuristic($\lambda = 1.0$)};
  \end{axis}
\end{tikzpicture}
  }
  \caption{Free path model on SWAN for workload FB, the different choice of time interval $\epsilon$ may affect the performance bound of time interval LP value and the performance of heuristic ($\lambda = 1$).}
  \label{fig:experiment_epsilon}
\end{figure*}
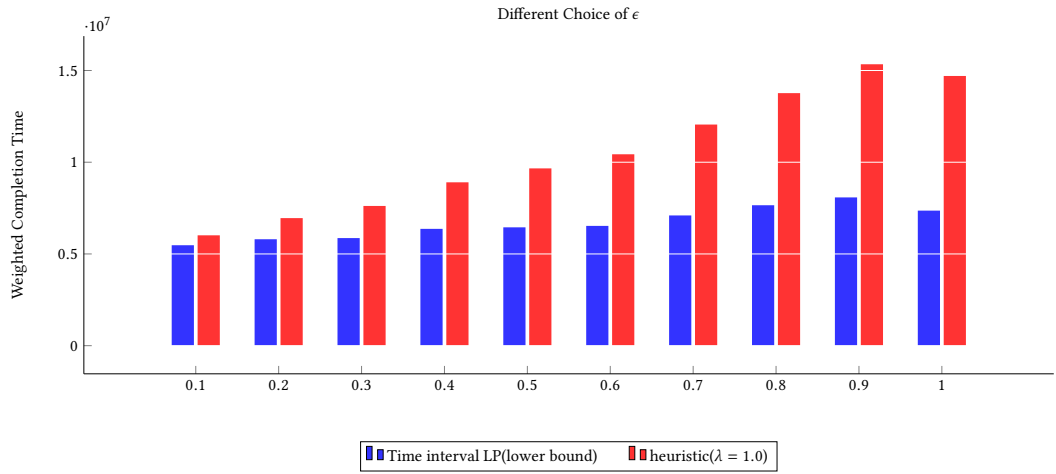
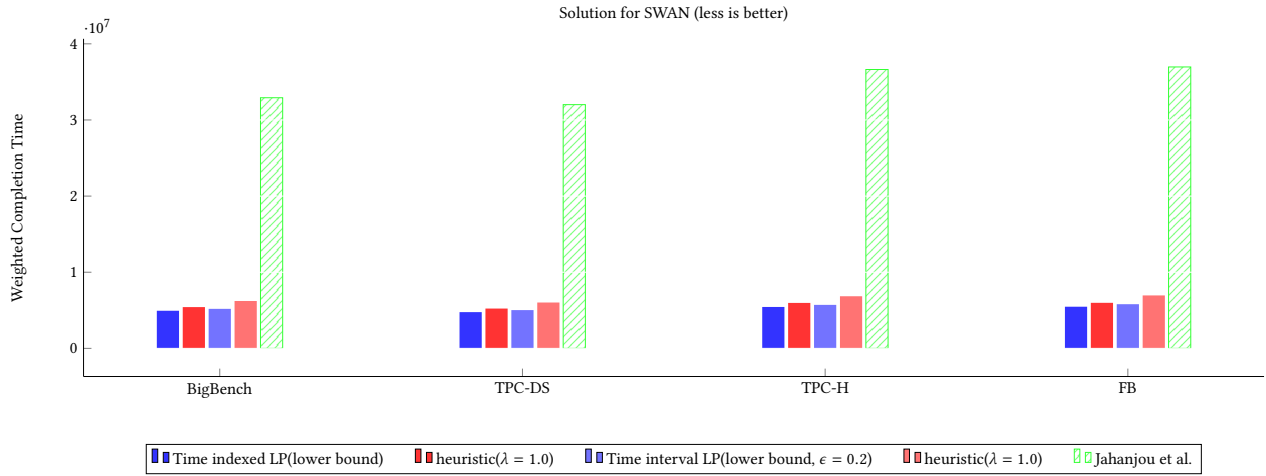
\begin{figure*}
  \centering
  \scalebox{0.7}{%
    \begin{tikzpicture}
 \centering
  \begin{axis}[
    ybar,
    enlargelimits=true,
    axis on top,
    title={Solution for SWAN (less is better)},
    height=8cm, width=24cm,
    bar width=12,
    ymajorgrids, tick align=inside,
    major grid style={draw=white},
    ymin=0,
    axis x line*=bottom,
    axis y line*=left,
    y axis line style={opacity=50},
    enlarge x limits=0.15,
    legend style={
      at={(0.5,-0.2)},
      anchor=north,
      legend columns=-1,
      /tikz/every even column/.append style={column sep=15}
    },
    ylabel={Weighted Completion Time},
    xtick=data,
    symbolic x coords={
      BigBench,
      TPC-DS,
      TPC-H,
      FB},
    ]

   \addplot [draw=none, fill=blue!80] coordinates {
      (BigBench, 4940095.478798)
      (TPC-DS, 4759818.526494) 
      (TPC-H, 5437756.440628)
      (FB, 5470255.409429) };
   \addplot [draw=none, fill=red!80] coordinates {
     (BigBench, 5420742.084466)
      (TPC-DS, 5229347.588657) 
      (TPC-H, 5970175.003998)
      (FB, 5982397.045652) };

    \addplot [draw=none, fill=blue!55] coordinates {
      (BigBench, 5191395.316109)
      (TPC-DS, 5024907.560099) 
      (TPC-H, 5712109.407223)
      (FB, 5801102.090175) };
    \addplot [draw=none, fill=red!55] coordinates {
      (BigBench, 6221959.970287)
      (TPC-DS, 6022303.185023) 
      (TPC-H, 6845889.620363)
      (FB, 6953022.289864) };


   \addplot [draw=green!80, pattern color=green!50, pattern=north east lines] coordinates {
     (BigBench,32919355.8031516)
      (TPC-DS, 32011858.728221342) 
      (TPC-H, 36627550.76485181)
      (FB, 36962633.67680418) };
    \legend{%
      Time indexed LP(lower bound), heuristic($\lambda = 1.0$), 
      Time interval LP({lower bound, $\epsilon = 0.2$}), heuristic($\lambda = 1.0$), 
      Jahanjou et al.}
  \end{axis}
\end{tikzpicture}
  }
  \caption{Single path model on SWAN, showing the performance bound of time indexed and time interval LP value, the performance of heuristic ($\lambda = 1$), best $\lambda$ among samples, and the expected value when $\lambda$ is chosen from the distribution mentioned in Section~\ref{sec:algorithm}. Here we compare against algorithm by Jahanjou et al.\cite{jahanjou2017asymptotically}.}
  \label{fig:experiment_path_swan}
\end{figure*}
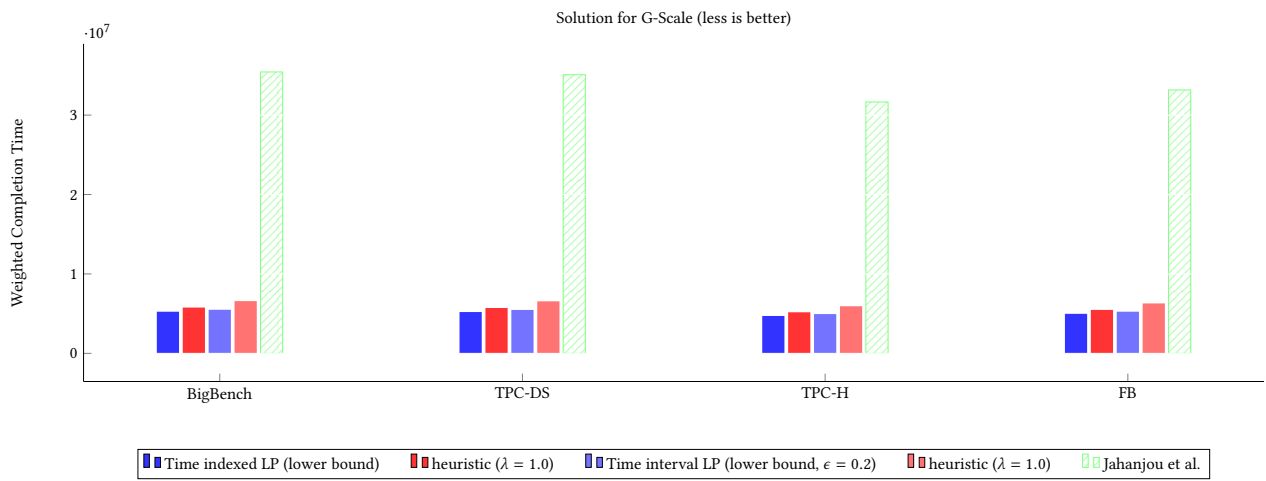
\begin{figure*}
  \centering
  \scalebox{0.7}{%
    \begin{tikzpicture}
 \centering
  \begin{axis}[
    ybar,
    enlargelimits=true,
    axis on top,
    title={Solution for G-Scale (less is better)},
    height=8cm, width=24cm,
    bar width=12,
    ymajorgrids, tick align=inside,
    major grid style={draw=white},
    ymin=0,
    axis x line*=bottom,
    axis y line*=left,
    y axis line style={opacity=50},
    enlarge x limits=0.15,
    legend style={
      at={(0.5,-0.2)},
      anchor=north,
      legend columns=-1,
      /tikz/every even column/.append style={column sep=15}
    },
    ylabel={Weighted Completion Time},
    xtick=data,
    symbolic x coords={
      BigBench,
      TPC-DS,
      TPC-H,
      FB},
    ]
   \addplot [draw=none, fill=blue!80] coordinates {
      (BigBench, 5242688.573051)
      (TPC-DS, 5202072.170825) 
      (TPC-H, 4707164.815620)
      (FB, 4976154.922722) };
   \addplot [draw=none, fill=red!80] coordinates {
      (BigBench, 5764974.093929)
      (TPC-DS, 5710642.526254) 
      (TPC-H, 5171114.886991)
      (FB, 5469704.722791) };

    \addplot [draw=none, fill=blue!55] coordinates {
      (BigBench, 5487285.812247)
      (TPC-DS, 5465266.501561) 
      (TPC-H, 4946192.074520)
      (FB, 5245794.918606) };
    \addplot [draw=none, fill=red!55] coordinates {
      (BigBench, 6576379.553527)
      (TPC-DS, 6550235.737915) 
      (TPC-H, 5928008.278105)
      (FB, 6287158.579652) };


   \addplot [draw=green!50, pattern color=green!30, pattern=north east lines] coordinates {
     (BigBench, 35432394.89525294)
      (TPC-DS, 35082610.34512873) 
      (TPC-H, 31652595.624498144)
      (FB, 33167098.53244567) };
    \legend{%
      Time indexed LP (lower bound), heuristic ($\lambda = 1.0$), 
      Time interval LP ({lower bound, $\epsilon = 0.2$}), heuristic ($\lambda = 1.0$), 
      Jahanjou et al.}
  \end{axis}
\end{tikzpicture}
  }
  \caption{Single path model on G-Scale, showing the performance bound of time indexed and time interval LP value, the performance of heuristic ($\lambda = 1$), best $\lambda$ among samples, and the expected value when $\lambda$ is chosen from the distribution mentioned in Section~\ref{sec:algorithm}. Here we compare against algorithm by Jahanjou et al.\cite{jahanjou2017asymptotically}.}
  \label{fig:experiment_path_gb4}
\end{figure*}
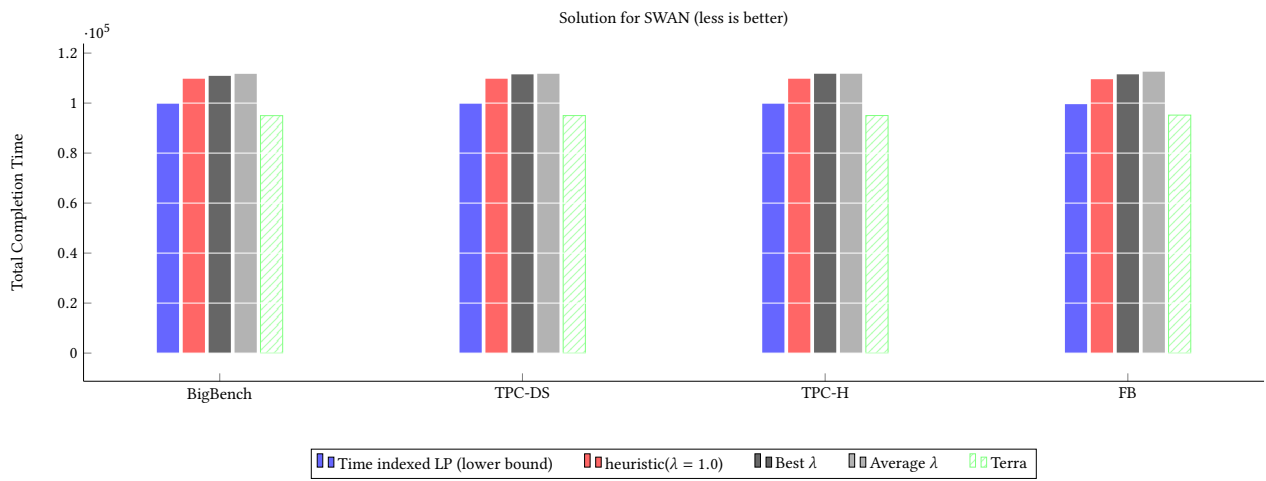
\begin{figure*}
  \centering
  \scalebox{0.7}{%
    \begin{tikzpicture}
 \centering
  \begin{axis}[
    ybar,
    enlargelimits=true,
    axis on top,
    title={Solution for SWAN (less is better)},
    height=8cm, width=24cm,
    bar width=12,
    ymajorgrids, tick align=inside,
    major grid style={draw=white},
    ymin=0,
    axis x line*=bottom,
    axis y line*=left,
    y axis line style={opacity=50},
    enlarge x limits=0.15,
    legend style={
      at={(0.5,-0.2)},
      anchor=north,
      legend columns=-1,
      /tikz/every even column/.append style={column sep=15}
    },
    ylabel={Total Completion Time},
    xtick=data,
    symbolic x coords={
      BigBench,
      TPC-DS,
      TPC-H,
      FB},
    ]
    \addplot [draw=none, fill=blue!60] coordinates {
      (BigBench, 99950.000000)
      (TPC-DS, 99950.000000) 
      (TPC-H, 99950.000000)
      (FB, 99576.088335) };
   \addplot [draw=none,fill=red!60] coordinates {
      (BigBench, 109750.000000)
      (TPC-DS, 109750.000000) 
      (TPC-H, 109750.000000)
      (FB, 109550.000000) };
   \addplot [draw=none, fill=black!60] coordinates {
     (BigBench, 110896.000000)
      (TPC-DS, 111496.000000) 
      (TPC-H, 111745.000000)
      (FB, 111510.000000) };
   \addplot [draw=none, fill=black!30] coordinates {
     (BigBench, 111702.550000 )
      (TPC-DS, 111732.550000) 
      (TPC-H, 111745.000000)
      (FB, 112561.750000) };

   \addplot [draw=green!50, pattern color=green!30, pattern=north east lines] coordinates {
     (BigBench, 94984.000000)
      (TPC-DS, 95011.000000) 
      (TPC-H, 95026.000000)
      (FB, 95197.000000) };
    \legend{%
      Time indexed LP (lower bound), heuristic($\lambda = 1.0$), Best $\lambda$, Average $\lambda$,
      Terra}
  \end{axis}
\end{tikzpicture}
  }
  \caption{Free path model with no weight on graph SWAN, showing the performance bound of time indexed LP value, the performance of heuristic ($\lambda = 1$), best $\lambda$ among samples, and the expected value when $\lambda$ is chosen from the distribution mentioned in Section~\ref{sec:algorithm}. Here we compare against Terra\cite{chowdhury2018terra}}
  \label{fig:experiment_flow_swan_no_weight}
\end{figure*}
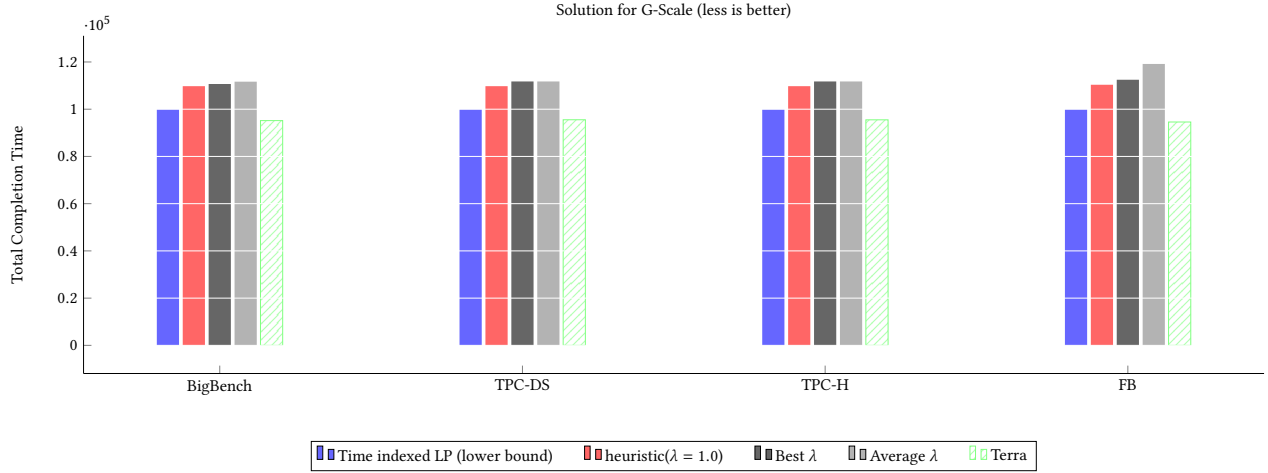
\begin{figure*}
  \centering
  \scalebox{0.7}{%
    \begin{tikzpicture}
 \centering
  \begin{axis}[
    ybar,
    enlargelimits=true,
    axis on top,
    title={Solution for G-Scale (less is better)},
    height=8cm, width=24cm,
    bar width=12,
    ymajorgrids, tick align=inside,
    major grid style={draw=white},
    ymin=0,
    axis x line*=bottom,
    axis y line*=left,
    y axis line style={opacity=50},
    enlarge x limits=0.15,
    legend style={
      at={(0.5,-0.2)},
      anchor=north,
      legend columns=-1,
      /tikz/every even column/.append style={column sep=15}
    },
    ylabel={Total Completion Time},
    xtick=data,
    symbolic x coords={
      BigBench,
      TPC-DS,
      TPC-H,
      FB},
    ]
    \addplot [draw=none, fill=blue!60] coordinates {
      (BigBench, 99950.000000)
      (TPC-DS, 99950.000000) 
      (TPC-H, 99950.000000)
      (FB, 99972.745998) };
   \addplot [draw=none,fill=red!60] coordinates {
      (BigBench, 109750.000000)
      (TPC-DS, 109750.000000) 
      (TPC-H, 109750.000000)
      (FB, 110350.000000) };
   \addplot [draw=none, fill=black!60] coordinates {
     (BigBench, 110671.000000)
      (TPC-DS, 111745.000000) 
      (TPC-H, 111745.000000)
      (FB, 112473.000000) };
   \addplot [draw=none, fill=black!30] coordinates {
     (BigBench, 111653.050000)
      (TPC-DS, 111745.000000) 
      (TPC-H, 111745.000000)
      (FB, 119162.550000) };

   \addplot [draw=green!50, pattern color=green!30, pattern=north east lines] coordinates {
     (BigBench, 95173.000000)
      (TPC-DS, 95528.000000) 
      (TPC-H, 95502.000000)
      (FB, 94579.000000) };
    \legend{%
      Time indexed LP (lower bound), heuristic($\lambda = 1.0$), Best $\lambda$, Average $\lambda$,
      Terra}
  \end{axis}
\end{tikzpicture}
  }
  \caption{Free path model with no weight on graph G-Scale, showing the performance bound of time indexed LP value, the performance of heuristic ($\lambda = 1$), best $\lambda$ among samples, and the expected value when $\lambda$ is chosen from the distribution mentioned in Section~\ref{sec:algorithm}. Here we compare against Terra\cite{chowdhury2018terra}}
  \label{fig:experiment_flow_gb4_no_weight}
\end{figure*}

\textbf{Impact of $\lambda$:} See Figure~\ref{fig:experiment_flow_swan} and Figure~\ref{fig:experiment_flow_gb4}. When $\lambda$ is $1.0$, we take the LP solution directly (this is exactly the LP-based heuristic). Across all experiments, this seems the best choice of $\lambda$. The best sampled $\lambda$ and the average case $\lambda$ are pretty close, indicating the performance does not change much across different $\lambda$.

\textbf{Impact of $\epsilon$:} To study the effect of the size of the time interval, we measure the LP objective and the schedule obtained by the LP-based heuristic as we vary $\epsilon$ in Figure~\ref{fig:experiment_epsilon}. As $\epsilon$ increases, the size of the linear program will drop, making it faster to solve. On the other hand, the quality of solution drops, as we will not start a job until the whole current interval is after its release time, and will not consider a job finished until the interval its completion time belongs to ends. Thus a proper selection of $\epsilon$ may depend on the available computational resources for solving the LP.

\textbf{Single Path Model:} 
Figures~\ref{fig:experiment_path_swan} and \ref{fig:experiment_path_gb4} compare the performance of our algorithms with that of Jahanjou et al.~\cite{jahanjou2017asymptotically} on all the benchmarks and topologies. Across all the experiments, we observe that our algorithms perform significantly better.

\textbf{Free Path Model:} See Figure~\ref{fig:experiment_flow_swan_no_weight} and Figure~\ref{fig:experiment_flow_gb4_no_weight} for comparisons with the algorithm in Terra\cite{chowdhury2018terra}. Since Terra only handles uniform coflow weights, we set all weights to be unit for these experiments.
Surprisingly, we observe that Terra performs slightly better than even the LP objective itself. This disparity arises as the LP relies on time slots of 50 seconds while Terra deals with time slots of much finer granularity. For the weighted case, we are not aware of previous work, and only compare to LP solution in Figures \ref{fig:experiment_flow_swan} and \ref{fig:experiment_flow_gb4}.

\section{Conclusion}
\label{sec:conclusion}

In this paper we developed an efficient approximation algorithm for the coflow scheduling problem in general graph topologies. This algorithm is shown to be practical and one that delivers extremely high quality solutions. The new insight was to write a time indexed LP formulation and to convert it using the idea of stretching the schedule. 

The next major challenge is developing {\em online} methods for coflow shceduling to minimize weighted flow time. Prior work \cite{khuller2018select} deals with the problem of minimizing weighted completion time by making use of offline approximation algorithms. However, the problem of minimizing weighted flow time is considerably more challenging. The technical difference is that flow time is defined as $C_j-r_j$ where $C_j$ is the completion time of a job, and $r_j$ the release time. Optimizing flow time non-preemptively even on a single machine (a different model) is a notoriously difficult problem with some recent progress \cite{batra2018constant,feige2018polynomial}.

%
\bibliographystyle{ACM-Reference-Format}
\balance
\bibliography{bibfile}

\ifjournal{%
\appendix
\section{Sketch of generalization to super-polynomial time span}
\label{sec:large_jobs}

Geometric series time interval is defined as follows. For an $\epsilon > 0$, let $\tau_0 = 0, \tau_1 = 1, \cdots,  \tau_k = (1 + \epsilon)^{k - 1},\cdots $. We define the $k$-th interval as $l_k = [\tau_{k-1}, \tau_k]$. Since $T$ is at most the sum of all processing time and all release time, we know the number of intervals $\mathbb{T} = 1 + \lceil \log_{1 + \epsilon}T \rceil$ is polynomial.

We change the LP as follows. We abuse notation a bit and allow $\mathbb{T}$ to represent the set $\{1, 2, \cdots, \mathbb{T}\}$ when there is no confusion. We replace all accurance of $T$ with $\mathbb{T}$ in Section~\ref{sec:lp_formulation}, modify Equation~(\ref{lp1:release time}) and Equation~(\ref{lp1:completion}) to accommodate for release time, and get the following linear program.

\begin{align}
  &\text{Minimize} \quad \sum_jw_jC_j\nonumber\\
  &\text{Subject to} \nonumber\\
  &\sum_tx_{j}^i(t) = 1 && \forall j \in [n], \forall i \in [n_j]\label{lp2:job finish}\\
&X_j(t) \leq \sum_{\ell = 1}^t x_j^i(\ell) && \forall j \in [n], \forall i \in [n_j], \forall t \in \mathbb{T} \label{lp2:cumulative}\\
&C_j \geq 1 + \sum_t (\tau_{t} - \tau_{t - 1})(1 - X_j(t)) && \forall j\in [n] \label{lp2:completion}\\
&r_j^i \geq \tau_t \Rightarrow x_{j}^i(t) = 0 && \forall j \in [n],\forall i \in [n_j], \forall t\in \mathbb{T}\label{lp2:release time}\\
&x_{j}^i(t) \geq 0 && \forall j \in [n],\forall i \in [n_j], \forall t\in \mathbb{T} \label{lp2:positive}
\end{align}

For the model specific part of linear program, we only need to change the capacity constraints: replace Equation~(\ref{eq:path_not_given_constraint}) for single path model to get
\begin{align}
  \sum_{p_j^i\ni e} x^i_j(t)\cdot \sigma^i_j &\leq (\tau_t - \tau_{t-1})c(e), && \forall e\in E, \forall t\in \mathbb{T}\label{eq:path_not_given_constraint_}
\end{align}
and Equation~(\ref{eq:capacity}) for free path model to get
\begin{align}
  &\sum_{e \in \delta_{out}(s_j^i)}x^i_j(t, e) = x^i_j(t), &&\forall j \in [n],\forall i \in [n_j], \forall t\in \mathbb{T}\label{eq:source_}\\
  &\sum_{e \in \delta_{in}(t_j^i)}x^i_j(t, e) = x^i_j(t), &&\forall j \in [n],\forall i \in [n_j], \forall t\in \mathbb{T}\label{eq:sink_}\\
  &\sum_{e \in \delta_{in}(v)}x^i_j(t, e) = \sum_{e = \delta_{out}(v)}x^i_j(t, e), &&\forall j \in [n],\forall i \in [n_j], \forall t\in \mathbb{T}, \nonumber\\
  &&&\ \forall v\in V\backslash \{s_i, t_i\}\label{eq:conservation_}\\
  &\sum_{j \in [n],i \in [n_j]}\hspace{-0.5cm} x^i_j(t, e)\cdot \sigma^i_j \leq (\tau_t - \tau_{t-1})c(e), &&\forall t\in \mathbb{T}, \forall e\in E\label{eq:capacity_}
\end{align}

Similar to Proposition~\ref{prop:lpfeasibility}, we prove Constraint~(\ref{lp2:completion}) is a good lower bound.

\begin{proposition}
\label{prop:lpfeasibility_}
The completion time of a coflow $F_j$ can be lower bounded by $C_j \geq 1 + \sum_t (\tau_{t} - \tau_{t - 1})(1 - X_j(t))$ where $X_j(t) \in [0,1]$ denotes the fraction of coflow $F_j$ that has been completed by (the end of) time interval $[\tau_{t - 1}, \tau_t]$.
\end{proposition}
\begin{proof}
  If a job completes in the interval $(\tau_{t-1}, \tau_t]$, then its finishing time is at least $\tau_{t-1} + 1$.
\begin{align*}
  C_j &\geq \sum_{t = 1}^{\mathbb{T}} (1 + \tau_{t-1}) \cdot x_j(t) = \sum_{t=1}^{\mathbb{T}}x_j(t) + \sum_{t = 1}^{\mathbb{T}}x_j(t)\sum_{\rho=1}^{t - 1}(\tau_{\rho} - \tau_{\rho-1})\\
      &= 1 + \sum_{\rho=1}^{\mathbb{T} - 1}(\tau_{\rho} - \tau_{\rho-1})\sum_{t = \rho + 1}^{\mathbb{T}}x_j(t) \\
      &= 1 + \sum_{\rho=1}^{\mathbb{T} - 1}(\tau_{\rho} - \tau_{\rho-1})\left(\sum_{t = 1}^{\mathbb{T}}x_j(t) - \sum_{t = 1}^{\rho}x_j(t) \right)\\
      &= 1 + \sum_{\rho=1}^{\mathbb{T} - 1}(\tau_{\rho} - \tau_{\rho-1})(1 - X_j(\rho))
\end{align*}
\end{proof}

After getting a solution, we would schedule coflows into intervals instead of into time slots. Inside each time interval, we just schedule each flow at uniform speed, and break into actual time slots. Similar to Section~\ref{sec:algorithm}, we can prove that this solution is feasible.

\subsection{Analysis}
\label{sec:analysis-random_}

Recall that $C_j^*$ denotes the completion time of the coflow $F_j$ in the optimal LP solution. For any continuous time $t \in [0, T]$, define $\widehat{X}_j(t)$ to be the fraction of coflow $F_j$ that has been scheduled in the LP solution by time $t$. Note $X_j(t)$ is for time interval $[\tau_{t-1}, \tau_{t}]$, but $\widehat{X}_j(t)$ is for original time slots. Flows are scheduled at an uniform rate in every time interval. Use $\rho(t)$ to denote the smallest $\rho$ such that $t\in (\tau_{\rho - 1}, \tau_{\rho}]$, we have
\begin{align}
\widehat{X}_j(t) = X_j(\rho(t)) + \frac{t - \tau_{\rho(t)- 1}}{\tau_{\rho(t)} - \tau_{\rho(t) - 1}} \left( X_j(\rho(t) + 1) - X_j(\rho(t)) \right)
 \label{eq:tau_}
\end{align}

Similar to Lemma~\ref{lem:integral}, we state and prove the following lemma.
\begin{lemma}
\label{lem:integral_}
$\int_{t = 0}^{T} (1 - \widehat{X}_j(t)) \mathbf{d} t \leq (1 + \epsilon) C_j^* - \frac{1}{2}$ where $X_j(\rho)$ is defined as per Eq. \eqref{eq:tau_}.
\end{lemma}

\begin{proof}
From constraints \eqref{lp2:completion}, we have that
\begin{align*}
  &\int_{t = 0}^{T} (1 - \widehat{X}_j(t)) \mathbf{d} t\\
  =&\sum_{\rho = 1}^{\mathbb{T}}\int_{t = \tau_{\rho - 1}}^{\tau_{\rho}} (1 - \widehat{X}_j(t)) \mathbf{d} t\\
  \intertext{Since $1 - \widehat{X}_j(t)$ is linear for $t\in (\tau_{\rho - 1}, \tau_{\rho}]$,}
  =&\sum_{\rho = 1}^{\mathbb{T}}\frac{\tau_{\rho} - \tau_{\rho - 1}}{2}(1 - X_j(\rho) + 1 - X_j(\rho - 1))\\
  =&\sum_{\rho = 1}^{\mathbb{T}}\frac{\tau_{\rho} - \tau_{\rho - 1}}{2}(1 - X_j(\rho)) + \sum_{\rho = 1}^{\mathbb{T}}\frac{\tau_{\rho} - \tau_{\rho - 1}}{2}(1 - X_j(\rho - 1))\\
  =&\sum_{\rho = 1}^{\mathbb{T}}\frac{\tau_{\rho} - \tau_{\rho - 1}}{2}(1 - X_j(\rho))\\
  &+ (1 + \epsilon)\sum_{\rho = 2}^{\mathbb{T}}\frac{\tau_{\rho - 1} - \tau_{\rho - 2}}{2}(1 - X_j(\rho - 1)) + \frac{\tau_1 - \tau_0}{2} (1 - X_j(0))\\
  =&\sum_{\rho = 1}^{\mathbb{T} - 1}\frac{\tau_{\rho} - \tau_{\rho - 1}}{2}(1 - X_j(\rho)) + \frac{\tau_{\mathbb{T}} - \tau_{\mathbb{T} - 1}}{2}(1 - X_j(\mathbb{T}))\\
  &+ (1 + \epsilon)\sum_{\rho = 1}^{\mathbb{T} - 1}\frac{\tau_{\rho} - \tau_{\rho - 1}}{2}(1 - X_j(\rho)) + \frac{\tau_1 - \tau_0}{2} (1 - X_j(0))\\
  =&\frac{2 + \epsilon}{2}\sum_{\rho = 1}^{\mathbb{T} - 1}(\tau_{\rho} - \tau_{\rho - 1})(1 - X_j(\rho)) + \frac{1}{2}\\
  \intertext{Plugging in Proposition~\ref{prop:lpfeasibility_},}
  \leq & \frac{2 + \epsilon}{2}(C_j^* - 1) + \frac{1}{2} \\
   \leq  & (1 + \epsilon) C_j^* -  \frac{1}{2} \\
\end{align*}
\end{proof}

For any $\lambda \in [0,1]$, define $C_j^*(\lambda)$ to be the earliest time $\rho$ such that $\lambda$ fraction of the coflow $F_j$ has been scheduled in the LP solution, i.e., in other words its the smallest $t$ such that $\widehat{X}_j(t) = \lambda$. Note that by time $C_j^*(\lambda)$, $\lambda$ fraction of \emph{every flow $f^i_j \in F_j$} has been scheduled by the LP.

\begin{proposition}
\label{prop:counting_}
$ \displaystyle \int_{\lambda = 0}^1 C_j(\lambda) \mathbf{d}\lambda \leq \int_{t = 0}^{T}(1 - \widehat{X}_j(t)) \mathbf{d} t$.
\end{proposition}

\begin{proof}
\begin{align*}
  \int_{\lambda = 0}^1 C_j(\lambda) \mathbf{d}\lambda
  &= \int_{\lambda = 0}^1 \int_{t = 0}^{T} \mathbbm{1}_{[ C_j(\lambda) > t]} \mathbf{d} t \mathbf{d}\lambda\\
  &= \int_{t = 0}^{T} \int_{\lambda = 0}^1 \mathbbm{1}_{[C_j(\lambda) > t]} \mathbf{d}\lambda \mathbf{d} t\\
  &= \int_{t = 0}^{T} \int_{\lambda = X_j(t)}^1 1 \mathbf{d}\lambda \mathbf{d} t\\
  &= \int_{t = 0}^{T}(1 - \widehat{X}_j(t)) \mathbf{d} t
\end{align*}
\end{proof}

Finally, we are ready to bound the completion time $C_j(alg)$ of coflow $F_j$ in the stretched schedule. For any fixed $\lambda \in (0,1)$, since we stretch the schedule by a factor of $\frac{1}{\lambda}$, it is easy to verify\footnote{All flows $f_j^i \in F_j$ were completed by at least $\lambda$ fraction by time $C_j(\lambda)$. So in the stretched schedule, all those flows must be completed by time $\left\lceil \frac{C_j(\lambda)}{\lambda} \right\rceil$.} that $C_j(alg) \leq \left\lceil \frac{C_j(\lambda)}{\lambda} \right\rceil$.
Since $\lambda$ is drawn randomly from a distribution, the following lemma bounds the expected completion time of coflow $F_j$ in the stretched schedule.

\begin{lemma}
The expected completion time of any coflow $F_j$ in the stretched schedule is bounded by $2(1 + \epsilon)C_j^*$.
\end{lemma}

\begin{proof}
\begin{align*}
\mathbb{E}[C_j(alg)] &\leq \int_{\lambda=0}^1 f(\lambda) \left\lceil  \frac{C_j(\lambda)}{\lambda} \right\rceil d\lambda\\
&\leq\int_{\lambda=0}^1 2 \lambda \left( 1 +\frac{C_j(\lambda)}{\lambda}  \right) d\lambda\\
&= \int_{\lambda=0}^1 2 \lambda d\lambda + 2\int_{\lambda=0}^1 C_j(\lambda) d\lambda\\
&= 1 + 2 \int_{\lambda=0}^1 C_j(\lambda) d\lambda
\intertext{By Lemma \ref{lem:integral_} and Proposition \ref{prop:counting_},}
&\leq 2 (1 + \epsilon)C_j^*.
\end{align*}
\end{proof}

Theorem \ref{thm:main_super_polynomial} thus follows from the linearity of expectation.

}\fi
\end{document}